\documentclass[11pt,journal]{IEEEtran}
\usepackage{amsmath,amsfonts,amssymb}
\usepackage{empheq}

\usepackage{amsthm,fancybox,hyperref}
\usepackage{inconsolata}
\usepackage{hyperref}
\usepackage{wrapfig}
\hypersetup{
	colorlinks=true,
	linkcolor=black,
	filecolor=black,      
	urlcolor=black,
	citecolor = blue,
}
\usepackage{algorithmic}
\usepackage{algorithm}
\usepackage{array}
\usepackage{booktabs}
\usepackage{cases}
\usepackage{comment}

\usepackage[font=footnotesize]{caption}
\usepackage{sidecap}
\sidecaptionvpos{figure}{t}

\usepackage{nicefrac,xfrac}
\usepackage{dsfont}
\usepackage{mathrsfs}
\usepackage{stmaryrd}
\usepackage[inline, shortlabels]{enumitem}

\setlist{itemjoin ={,\enspace},itemjoin* = { and,\enspace}}
\newcommand{\sqenumi}{\protect\fbox{\small\arabic*}}

\usepackage[many,breakable]{tcolorbox}
\usetikzlibrary{calc,shadows.blur}
\newtcolorbox{gbox}{boxsep=3pt,left=3pt,right=3pt,top=3pt,bottom=3pt,
	colback=red!2,colframe=red!90!black!20,arc=0mm,
	breakable}

\usepackage{url}
\usepackage{pifont}
\usepackage{verbatim}
\usepackage{graphicx}

\usepackage{cite}
\usepackage{xspace}
\usepackage{tikz}
\usepackage{colortbl}
\usepackage{theoremref}
\usepackage[normalem]{ulem}
\usepackage{xfrac}
\usepackage{soul}

\usepackage{flushend}

\definecolor{shadowcolor}{rgb}{0,.5,.5}
\DeclareFontFamily{U}  			{MnSymbolA}{}
\DeclareFontFamily{U}  			{MnSymbolC}{}
\DeclareSymbolFont{MnSyA}         	{U}  {MnSymbolA}{m}{n}
\DeclareSymbolFont{MnSyC}         	{U}  {MnSymbolC}{m}{n}

\DeclareFontShape{U}{MnSymbolA}{m}{n}{
	<-6>  MnSymbolA5
	<6-7>  MnSymbolA6
	<7-8>  MnSymbolA7
	<8-9>  MnSymbolA8
	<9-10> MnSymbolA9
	<10-12> MnSymbolA10
	<12->   MnSymbolA12}{}

\DeclareFontShape{U}{MnSymbolC}{m}{n}{
	<-6>  MnSymbolC5
	<6-7>  MnSymbolC6
	<7-8>  MnSymbolC7
	<8-9>  MnSymbolC8
	<9-10> MnSymbolC9
	<10-12> MnSymbolC10
	<12->   MnSymbolC12}{}

\DeclareMathSymbol{\hpup}{\mathrel}{MnSyA}{56}
\DeclareMathSymbol{\hpdn}{\mathrel}{MnSyA}{50}

\def\Z								   {{\mathbb Z}}
\def\Zp								   {{{\mathbb Z}^+}}
\def\R								   {{\mathbb R}}
\def\C								   {{\mathbb C}}
\def\D                                 {{\mathcal D}}

\def\MSEhndiff                          {\MSEOp_\Delta\rob{\osfact, \nvec{n}}}

\def\ob								  {1\nobreakdash-Bit\xspace}

\def\BOmFT						  {{\mathcal{B}}_{\OmFT}}
\def\BOmLCT						 {\mathcal{B}_{\Omega_{\mat{\Lambda}}}}

\def\ck								   {c_k}

\def\conv							 {\ast}

\def\csgn							 {\mathsf{csgn}}

\def\DE								  {\stackrel{\rm{def}}{=}}

\def\diffop							 {\Delta}
\def\diffflt						   {v}
\def\dnL							 {m_{\Lmtx}^*}

\def\DFT							 {DFT\xspace}
\def\DTLCT						   {DT\nobreakdash-LCT\xspace}
\def\DLCT							{DLCT\xspace}

\def\IDLCT							 {IDLCT\xspace}

\def\IDTLCT						   {IDT\nobreakdash-LCT\xspace}
\def\idker							  {\varphi}
\def\ind							   {\mathds{1}}

\def\Fr									{FrT\xspace}
\def\FrFT							  {FrFT\xspace}
\def\FT								   {FT\xspace}
\def\gn								   {g\n}
\def\grec							  {\tilde{g}}
\def\HDR							 {HDR\xspace}
\def\O								   {\mathcal{O}}

\def\qL							  	  {q_{\Lmtx}}
\def\qLL							  {q_{\Lmtx,\lambda}}
\def\qLn							  {\qL\n}
\def\qLLn							  {\qLL\n}

\def\qLsecn					  		{\qL^{[2]}\n}
\def\qLLsecn					  	{\qLL^{[2]}\n}
\def\qMB							 {q_{\sf{MB}}}
\def\qMBn							{\qMB\n}

\def\qn								  {q\n}
\def\kL								  {{\kappa_\Lmtx}}

\def\kLinv						 	 {{\kappa_\Linvmtx}}
\def\KL							 	  {K_{\Lmtx}}

\def\L								   {{\mathcal{L}_{\Lmtx}}}
\def\Linv						     {{\mathcal{L}_{\Linvmtx}}}
\def\Lmtx					 	    {\mat{\Lambda}}
\def\LmtxFT					 	   {{\mat{\Lambda}_{\textsf{FT}}}}
\def\Linvmtx					   {{\mat{\Lambda}^{-1}}}
\def\LCS							 {LCS\xspace}
\def\LCT							 {LCT\xspace}
\def\LCTconv					 {\conv_{\Lmtx}}

\def\LFT							 {{\boldsymbol{\Lambda}_{\sf{FT}}}}

\def\LM								 {$\mathbf{\Lambda \Mmod}$\xspace}

\def\l									{\left (}
\def\m								  {\sqb{m}}
\def\Mmod						  {{\mathscr M}}
\def\Mc                             {\mathsf{M}_{\res}^{\nvec{c}}}
\def\Mn                             {\mathsf{M}_{\res}^{\nvec{n}}}
\def\nk								  {n_k}
\def\nosh							{\textit{noise shaping}\xspace}

\def\OmL						    {\Omega_{\Lmtx}}
\def\OmFT						   {\Omega_{\textsf{FT}}}
\def\osfact							{h}
\def\osrat							 {h}
\def\r									{\right )}
\def\n									{\sqb{n}}
\def\un								  {u\sqb{n}}
\def\upL							 {m_{\Lmtx}}

\def\USF							 {USF\xspace}

\def\res							   {\varepsilon}

\def\SD								  {\Sigma\Delta}
\def\SDQ							{$\Sigma\Delta\mathrm{Q}$\xspace}
\def\LSDQ							{$\Lambda\Sigma\Delta\mathrm{Q}$\xspace}
\def\LSDQone						{$\Lambda\Sigma\Delta\mathrm{Q}^\text{[1]}$\xspace}
\def\LSDQtwo						{$\Lambda\Sigma\Delta\mathrm{Q}^\text{[2]}$\xspace}
\def\sgn							 {\text{sgn}}
\def\sinc							  {\mathrm{sinc}}
\def\xn									{x\n}

\def\t									{\rob{t}}

\def\vL								   {{v_{\Lmtx}}}
\def\vLn							  {\vL\n}
\def\w 								   {\rob{\omega}}
\def\bsls							  {\boldsymbol{\Lambda}_{\sf S}}

\def\gt								    {g\t}
\def\MADC							{{${\Mmod}_{\lambda}$-{ADC}}\xspace}
\def\MSEOp						   {\mathbf{\epsilon}}

\def\Tcur								{T \osrat}
\def\MLSDQ						{$\mathscr{M}$\LSDQ}
\def\MLSDQone						{$\mathscr{M}\Lambda\Sigma\Delta\mathrm{Q}^\text{[1]}$\xspace}
\def\MLSDQtwo						{$\mathscr{M}\Lambda\Sigma\Delta\mathrm{Q}^\text{[2]}$\xspace}

\def\SNR                        {\mathsf{SNR}}
\def\Unif                       {\mathsf{Unif}}
\def\rhM                        {\rob{\osrat, \Mn}}
\def\rnc                        {\rob{\Mn, \Mc}}
\def\etal					{\emph{et al.}\xspace}

\def\ie					{\emph{i.e.}\xspace}

\newcommand\abs[1]				   		{\left | #1 \right |}

\newcommand\bpara[1]			  	   {\smallskip \noindent {\bf #1}}
\newcommand\dpara[1]			  	   {\smallskip \noindent \ding{224}}
\newcommand\MSElnc[1]               {\bar{\MSEOp}_{#1}\rob{\nvec{n}, \nvec{c}}}
\newcommand\MSElhn[1]               {\bar{\MSEOp}_{#1}\rob{\osfact, \nvec{n}}}
\newcommand\diff[1]				  		  {\underline{#1}}
\newcommand\diam[1]						{{#1}^{\Diamond}}
\newcommand\defref[1]				   {Definition \ref{#1}}
\newcommand\figref[1]					 {Fig. \ref{#1}}
\newcommand\secref[1]					 {Section \ref{#1}}
\newcommand\fracprt[1]			  	    {\left\llbracket #1 \right\rrbracket}
\newcommand\infnorm[1]				  {\norm{#1}_\infty}
\newcommand\LCTop[1]				  {\hat{#1}_{\Lmtx}}
\newcommand\LCTOpFT[1]				{\hat{#1}_{\LmtxFT}}

\newcommand\mat[1]				   		{\mathbf{#1}}
\newcommand\MSE[2]						{\MSEOp\rob{#1, #2}}
\newcommand\norm[1]				 		{\left\lVert#1\right\rVert}
\newcommand\rob[1]				    	 {\l #1 \r}

\newcommand\sqb[1]				 		 {\left [ #1 \right ]}
\newcommand\tabref[1]					{Table~\ref{#1}}
\newcommand\nvec[1]						{{{\mathbf{#1}}}}
\newcommand\uparr[1]					{ {\overset{\lower0.5em\hbox{$\smash{\scriptstyle	\hpup}$}} 	{{#1}}}}
\newcommand\dnarr[1]					{ {\overset{\lower0.5em\hbox{$\smash{\scriptstyle  	\hpdn}$}} 	{{#1}}}}

\newcommand{\EQc}[1]		{\stackrel{\eqref{#1}}{=}}

\renewcommand \hat \widehat

\renewcommand \tilde \widetilde
\renewcommand \leq \leqslant
\renewcommand \geq \geqslant

\renewcommand \bar \overline

\newtheorem{definition}{Definition}
\newtheorem{lemma}{Lemma}
\newtheorem{proposition}{Proposition}
\newtheorem{theorem}{Theorem}

\onecolumn

\begin{document}
    
	\title{1-Bit Unlimited Sampling Beyond Fourier Domain: \\
		Low-Resolution Sampling of Quantization Noise}
	
	\author{Václav Pavlíček and Ayush Bhandari
        \thanks{This work is supported by the UKRI’s HASC Program under grant EP/X040569/1, European Research Council’s Starting Grant for “CoSI-Fold” under grant 101166158 and the UKRI Future Leader’s Fellowship “Sensing Beyond Barriers via Non-Linearities” under grant MR/Y003926/1. Further details on {Unlimited Sensing} and materials on \textit{reproducible research} are available via  \href{https://bit.ly/USF-Link}{\texttt{https://bit.ly/USF-Link}}. We acknowledge computational resources and support provided by the Imperial College Research Computing Service (\href{http://doi.org/10.14469/hpc/2232}{\texttt{http://doi.org/10.14469/hpc/2232}}).}
		\thanks{The authors are with the Dept. of Electrical and Electronic Engineering, Imperial College London, South Kensington, London SW7 2AZ, UK. (Email: \texttt{\{vaclav.pavlicek20,a.bhandari\}@imperial.ac.uk}.}
		\thanks{Manuscript received Mon XX, 20XX; revised Mon XX, 20XX.}
	}

	\markboth{Accepted to IEEE-JSTSP, August~2025}%
	{}

	\maketitle
	
	\vspace{5mm}
	{\color{blue} 
	
	\centering 
	
	Accepted to IEEE Journal of Selected Topics in Signal Processing. \\

	\textit{Special Issue on Low-Bit-Resolution Signal Processing: Algorithms, Implementations, and Applications}
	
	}

		\vspace{10mm}

	\begin{abstract} 
		Analog-to-digital converters (ADCs) play a critical role in digital signal acquisition across various applications, but their performance is inherently constrained by sampling rates and bit budgets. This bit budget imposes a trade-off between dynamic range (DR) and digital resolution, with ADC energy consumption scaling linearly with sampling rate and exponentially with bit depth. To bypass this, numerous approaches, including oversampling with low-resolution ADCs, have been explored. A prominent example is 1-Bit ADCs with Sigma-Delta Quantization (SDQ), a widely used consumer-grade solution. However, SDQs suffer from overloading or saturation issues, limiting their ability to handle inputs with arbitrary DR. The Unlimited Sensing Framework (USF) addresses this challenge by injecting modulo non-linearity in hardware, resulting in a new digital sensing technology. In this paper, we introduce a novel 1-Bit sampling architecture that extends both conventional 1-Bit SDQ and USF. Our contributions are twofold: (1) We generalize the concept of noise shaping beyond the Fourier domain, allowing the inclusion of non-bandlimited signals in the Fourier domain but bandlimited in alternative transform domains. (2) Building on this generalization, we develop a new transform-domain recovery method for 1-Bit USF. When applied to the Fourier domain, our method demonstrates superior performance compared to existing time-domain techniques, offering reduced oversampling requirements and improved robustness. Extensive numerical experiments validate our findings, laying the groundwork for a broader generalization of 1-Bit sampling systems.
	\end{abstract}
	
	
	\newpage
	
	\tableofcontents	
	\newpage
	
	\section{Introduction}
	\IEEEPARstart{D}igital acquisition of signals is based on quantization in both time and amplitude dimensions. The well-known Shannon-Nyquist sampling theorem states that, for band-limited signals, time quantization is \emph{lossless} if the sampling rate exceeds a specific threshold, the Nyquist rate. This theorem serves as a bridge between the continuous and discrete realms. However, to fully digitize a signal, amplitude quantization is also necessary. Unlike time quantization, amplitude quantization introduces a permanent loss of information.
	
	In practice, digitization is performed by analog-to-digital converters (ADCs), which carry out both time and amplitude quantization. In this process, ``bits'' act as the digital currency, defining the precision of signals and controlling the extent of information loss during the amplitude quantization. Since ADCs are characterized by their sampling rates and bit budgets, their operation naturally involves fundamental trade-offs.

	\bpara{Trade-off \#1: Sampling Rate vs. Bit Budget.} From a practical perspective, ADCs consume energy that scales linearly with the sampling rate but exponentially \cite{Walden:2002:J} with the number of bits. Therefore, it is far more energy efficient to oversample than to increase the bit depth. This insight has driven the development of techniques that leverage oversampling to improve digital resolution within a fixed bit budget. Notable examples include oversampled ADCs \cite{Thao:1994:J,Verreault:2024:J} and dithering \cite{Schuchman:1964:J}, which redistribute quantization noise to enhance signal fidelity. A notable technological achievement in this context is \emph{one-bit} or {\ob sampling}, which comes in 3 main flavors:
	\begin{enumerate}[label = \roman*)]
		\item Sigma-Delta ($\SD$) Quantization or \SDQ \cite{Inose:1963:J,Aziz:1996:J,Daubechies:2003:J}
		\item Time-Encoded Sampling (or Asynchronous \SDQ) \cite{Lazar:2004:J}
		\item Sign-Based Sampling \cite{Shamai:1994:J,BarShalom:2002:J}.
	\end{enumerate}
	In all these variants, the core idea is to utilize simple, low-complexity hardware combined with high oversampling. This approach achieves effective resolution, even when the ADC resolution is reduced to a single bit.

\begin{wrapfigure}[24]{r}{0.65\textwidth}
 \begin{center}
 \vspace{-4mm}
\includegraphics[width=0.52\textwidth]{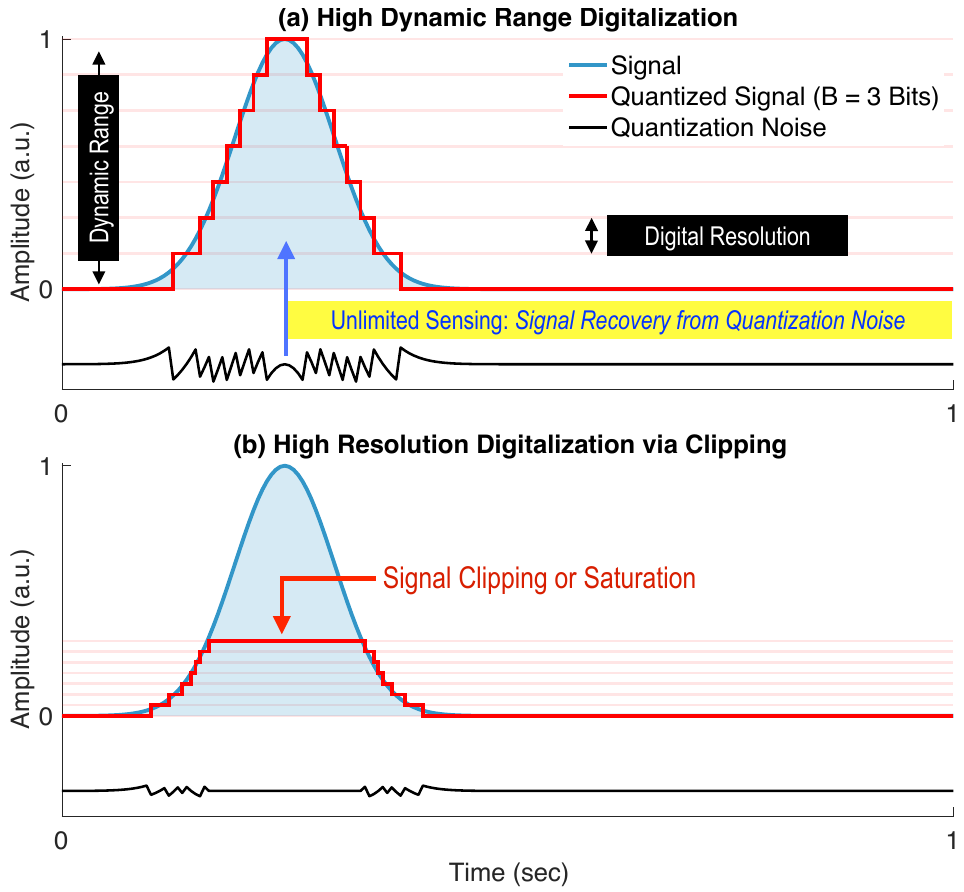}
		\caption{For a fixed bit budget ($3$ bits), the ADCs introduce a trade-off between dynamic range (DR) and digital resolution (DRes). (a) Covering full DR places a limit on the DRes, causing weak signals to be omitted. (b) Higher DRes can be achieved by signal clipping but this results in permanent loss of information.}
\label{fig:tradeoff}
\end{center}
\end{wrapfigure}

\bpara{Trade-off \#2: {Dynamic Range} (DR) vs. {Digital Resolution} (DRes)}. Given a fixed bit budget, it is impossible to simultaneously cover a signal's full dynamic range (DR) and achieve arbitrary digital resolution (DRes). The DR sets a fundamental limit on the voltage quantum, restricting the granularity of amplitude quantization, see \figref{fig:tradeoff}(a). Conversely, increasing the DRes is possible, but only at the expense of clipping the signal's range, see \figref{fig:tradeoff}(b). Both scenarios are suboptimal. In the first case,  signals weaker than the ADC's DRes are lost in the quantizer's null space. In the second case, signal saturation occurs due to range clipping, resulting in loss of information, particularly in pulse-like features. Even when bits are allocated to span the full DR, it has been widely reported \cite{AlSafadi:2012:J,Olofsson:2005:J,Sabharwal:2014:J,Zhang:2016:J} that ADCs remain vulnerable to saturation when high dynamic range (HDR) inputs exceed the ADC's designated DR.
	
	\bpara{Overcoming Trade-offs with Unlimited Sensing.} In recent years, the Unlimited Sensing Framework (USF) \cite{Bhandari:2017:C,Bhandari:2020:Ja,Bhandari:2021:J,Bhandari:2022:J} has emerged as an alternative digital acquisition pipeline capable of overcoming the bottlenecks of conventional ADCs. The USF is built upon a simple yet powerful mathematical insight: \emph{For smooth signals, their fractional part encodes their integer part.} This insight leads to a fundamentally new perspective on sampling: the quantized values of a signal (integer part) can be decoded from its quantization noise (fractional part) \cite{Bhandari:2020:C}. 
	
	Note that the quantization noise (QN)---equivalently, the fractional part or modulo representation---of a bandlimited signal manifests as a non-bandlimited signal (see \figref{fig:tradeoff}(a)). However, surprisingly, the sampling theorem at the core of the USF proves that constant-factor oversampling is sufficient to recover the bandlimited input from the QN or modulo samples. This ensures that time quantization in the USF is \uline{lossless}, akin to the Shannon-Nyquist theorem. Furthermore, with side information, recovery at the Nyquist rate is also possible \cite{Romanov:2019:J}.
	
	The subtlety of the USF lies in the fact that the QN must be acquired in the \emph{analog domain}, rather than in the digital domain, where it is traditionally interpreted. To highlight this distinction, we refer to the analog-domain QN as the \emph{modulo representation}, emphasizing the critical role of non-linearity applied prior to sampling and quantization. Modulo folding is achieved through innovative hardware implementation, resulting in \MADC (modulo ADC) where $\lambda>0$ is the ADC's DR. A variety of hardware validations are provided in \cite{Bhandari:2021:J,Bhandari:2022:J,Florescu:2022:J,Florescu:2022:Ja,Shtendel:2023:J,Guo:2024:J,Zhu:2024:C,Mulleti:2023:J}. The presence of modulo non-linearity prior to the ADC ensures that HDR signals are folded or aliased back into the ADC's DR, preventing saturation. As illustrated in \figref{fig:tradeoff}(a), for a given bit budget, the QN or modulo signal provides significantly higher DRes \cite{Bhandari:2021:J,Ordentlich:2018:J}.
	
\bpara{1-Bit USF.} Similar to conventional ADCs \cite{Verreault:2024:J}, \ob variants \cite{Inose:1963:J,Aziz:1996:J,Daubechies:2003:J,Lazar:2004:J,Shamai:1994:J,BarShalom:2002:J} impose a DR limit on the input signal; exceeding this limit leads to encoding errors and permanent information loss. This defect is overcome by leveraging USF as it folds the signal prior to entering the \ob ADC. While still in its early stages, the USF has been adapted to \ob sampling scheme, and this comes in three flavors (see \figref{fig:usf1bitschemes}).
	\begin{enumerate}[leftmargin=30pt, label=\sqenumi, itemsep = 5pt]
		\item \SDQ: Graf \etal \cite{Graf:2019:C} propose a \SDQ scheme to encode \ob modulo samples of a bandlimited signal. A \uline{time-domain thresholding} algorithm is then devised for residue estimation, which leads to recovery. The authors considered only noise-free modulo signal, and the study of robustness in the presence of noise is missing.
		\item Time-Encoding: Leveraging modulo hysteresis architecture \cite{Florescu:2022:J}, Florescu \& Bhandari \cite{Florescu:2022:Ja} develop recovery methods for \ob time-encoded \cite{Lazar:2004:J} modulo measurements, validated through simulations and hardware implementation. Extensions include system identification \cite{Florescu:2022:Ca}.
		\item Sign-Based \ob Sampling: Recently, Emaz \etal proposed UNO ({Unlimited One-Bit Sampling}) \cite{Eamaz:2024:J,Eamaz:2023:C} where modulo folded signal is encoded with a  sign-based quantizer, resulting in a \ob scheme. Emaz \etal consider a dithered threshold scheme combined with randomized Kaczmarz algorithm for signal recovery. The authors also consider sparse signals in their later papers \cite{Eamaz:2023:Ca}.
	\end{enumerate}
	
	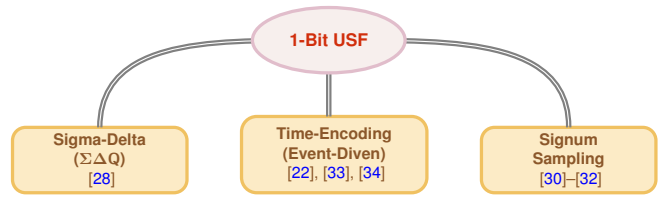
\begin{figure}[!t]
		\centering
		{\fontfamily{\sfdefault}\selectfont \scalebox{0.55}{

			\tikzset{every picture/.style={line width=0.75pt}} 
			
			\begin{tikzpicture}[x=0.75pt,y=0.75pt,yscale=-1,xscale=1]
				
				\draw  [color={rgb, 255:red, 240; green, 193; blue, 97 }  ,draw opacity=1 ][fill={rgb, 255:red, 252; green, 235; blue, 198 }  ,fill opacity=1 ][line width=2.25]  (240,124) .. controls (240,116.27) and (246.27,110) .. (254,110) -- (396,110) .. controls (403.73,110) and (410,116.27) .. (410,124) -- (410,166) .. controls (410,173.73) and (403.73,180) .. (396,180) -- (254,180) .. controls (246.27,180) and (240,173.73) .. (240,166) -- cycle ;
				\draw  [color={rgb, 255:red, 240; green, 193; blue, 97 }  ,draw opacity=1 ][fill={rgb, 255:red, 252; green, 235; blue, 198 }  ,fill opacity=1 ][line width=2.25]  (460,132) .. controls (460,125.37) and (465.37,120) .. (472,120) -- (608,120) .. controls (614.63,120) and (620,125.37) .. (620,132) -- (620,168) .. controls (620,174.63) and (614.63,180) .. (608,180) -- (472,180) .. controls (465.37,180) and (460,174.63) .. (460,168) -- cycle ;
				\draw  [color={rgb, 255:red, 240; green, 193; blue, 97 }  ,draw opacity=1 ][fill={rgb, 255:red, 252; green, 235; blue, 198 }  ,fill opacity=1 ][line width=2.25]  (30,132) .. controls (30,125.37) and (35.37,120) .. (42,120) -- (178,120) .. controls (184.63,120) and (190,125.37) .. (190,132) -- (190,168) .. controls (190,174.63) and (184.63,180) .. (178,180) -- (42,180) .. controls (35.37,180) and (30,174.63) .. (30,168) -- cycle ;
				\draw [color={rgb, 255:red, 128; green, 128; blue, 128 }  ,draw opacity=1 ][line width=1.5]    (108.5,119.99) .. controls (108.67,90.6) and (121.86,71.83) .. (142.84,59.81) .. controls (173.8,42.08) and (222.03,39.26) .. (270.29,38.75) .. controls (278.73,38.66) and (287.17,38.64) .. (295.52,38.62) .. controls (302.84,38.61) and (310.09,38.59) .. (311.31,38.57)(111.5,120.01) .. controls (111.66,91.89) and (124.25,73.91) .. (144.33,62.42) .. controls (174.94,44.89) and (222.64,42.25) .. (270.32,41.75) .. controls (278.75,41.66) and (287.19,41.64) .. (295.53,41.62) .. controls (302.86,41.61) and (310.11,41.59) .. (311.33,41.57) ;
				\draw [shift={(320,40)}, rotate = 179.38] [color={rgb, 255:red, 128; green, 128; blue, 128 }  ,draw opacity=1 ][line width=1.5]    (14.21,-4.28) .. controls (9.04,-1.82) and (4.3,-0.39) .. (0,0) .. controls (4.3,0.39) and (9.04,1.82) .. (14.21,4.28)   ;
				\draw [color={rgb, 255:red, 128; green, 128; blue, 128 }  ,draw opacity=1 ][line width=1.5]    (538.5,119.99) .. controls (538.5,119.82) and (538.5,119.64) .. (538.5,119.47) .. controls (538.5,92.55) and (526.43,74.99) .. (506.86,63.54) .. controls (474.36,44.53) and (421.54,42.13) .. (368.77,41.72) .. controls (361.16,41.66) and (353.55,41.64) .. (346,41.62) .. controls (337.22,41.61) and (328.51,41.59) .. (319.98,41.5)(541.5,120.01) .. controls (541.5,119.83) and (541.5,119.65) .. (541.5,119.47) .. controls (541.5,91.29) and (528.85,72.93) .. (508.37,60.95) .. controls (475.52,41.73) and (422.17,39.13) .. (368.79,38.72) .. controls (361.18,38.66) and (353.56,38.64) .. (346.01,38.62) .. controls (337.23,38.61) and (328.54,38.59) .. (320.02,38.5) ;
				\draw [color={rgb, 255:red, 128; green, 128; blue, 128 }  ,draw opacity=1 ][line width=1.5]    (321.5,40) -- (321.5,110)(318.5,40) -- (318.5,110) ;
				\draw  [color={rgb, 255:red, 225; green, 189; blue, 203 }  ,draw opacity=1 ][fill={rgb, 255:red, 247; green, 239; blue, 237 }  ,fill opacity=1 ][line width=2.25]  (250,40) .. controls (250,23.43) and (281.34,10) .. (320,10) .. controls (358.66,10) and (390,23.43) .. (390,40) .. controls (390,56.57) and (358.66,70) .. (320,70) .. controls (281.34,70) and (250,56.57) .. (250,40) -- cycle ;
				
				\draw (110,150) node  [font=\normalsize,color={rgb, 255:red, 139; green, 87; blue, 42 }  ,opacity=1 ] [align=left] {\begin{minipage}[lt]{99.65pt}\setlength\topsep{0pt}
						\begin{center}
							\textbf{Sigma-Delta ($\boldsymbol{\Sigma\Delta\mathsf{Q}}$)}\\\cite{Graf:2019:C}
						\end{center}
						
				\end{minipage}};
				\draw (325,145) node  [font=\normalsize,color={rgb, 255:red, 139; green, 87; blue, 42 }  ,opacity=1 ] [align=left] {\begin{minipage}[lt]{85pt}\setlength\topsep{0pt}
						\begin{center}
							\textbf{Time-Encoding }\\\textbf{(Event-Diven)}\\\cite{Florescu:2021:C, Florescu:2022:C, Florescu:2022:Ja}
						\end{center}
						
				\end{minipage}};
				\draw (540,150) node  [font=\normalsize,color={rgb, 255:red, 139; green, 87; blue, 42 }  ,opacity=1 ] [align=left] {\begin{minipage}[lt]{88.86pt}\setlength\topsep{0pt}
						\begin{center}
							\textbf{Signum Sampling}\\\cite{Eamaz:2023:C, Eamaz:2023:Ca, Eamaz:2024:J}
						\end{center}
						
				\end{minipage}};
				\draw (320,40) node  [font=\large,color={rgb, 255:red, 209; green, 46; blue, 12 }  ,opacity=1 ] [align=left] {\textbf{1-Bit USF}};

			\end{tikzpicture}
		}}
		\caption{Classification of 1-Bit USF schemes in literature.}
		\label{fig:usf1bitschemes}
	\end{figure}
	
	\bpara{Motivation and Technical Challenges.} The emergence of both \ob sampling \cite{Verreault:2024:J} and the USF, along with their combination \cite{Graf:2019:C,Florescu:2022:Ja,Eamaz:2024:J}, has led to significant advancements. However, two key challenges still hinder practical applications.  
	
	Firstly, while conventional \ob \SDQ \cite{Verreault:2024:J} leverages noise shaping in the Fourier domain \cite{Daubechies:2003:J} and has enabled successful implementations \cite{Verreault:2024:J}, there remain many signal classes that are non-bandlimited in the Fourier domain. Importantly, such signals can be bandlimited in other transform domains, such as the Fresnel transform (used in holography \cite{Gori:1981:J}) and the fractional Fourier transform (advantageous for MIMO systems \cite{Martone:2001:J, Rou:2024:J}). Generalizing the concept of bandwidth and consequently developing a noise-shaping approach beyond the Fourier domain would expand the applicability of \SDQ to these signal classes, thereby motivating the need for a generalized \SDQ framework. An initial attempt in this direction has been
    \cite{Bhandari:2020:FrFT}.
    Secondly, while the combination of \ob sampling \cite{Graf:2019:C,Florescu:2022:Ja,Eamaz:2024:J} and USF has been instrumental in addressing the dynamic range (DR) problem, current approaches operate exclusively in the time domain. As a result, they cannot handle signals specified by their Fourier features, such as band-pass and multi-band signals. To overcome this limitation, a transform-domain understanding of noise shaping is essential. However, this theoretical investigation is challenging due to the presence of modulo non-linearity. These aspects constitute the central motivation of this paper.
	
    \bpara{Contributions.} This paper introduces a {\bf generalized noise-shaping framework} that advances in two key directions: \uline{extending bandwidth} and \uline{overcoming dynamic range} limitations present in conventional approaches. Our main contributions are as follows:
	\begin{enumerate}[leftmargin = 35pt, label = $\textrm{C}_{\arabic*})$, itemsep = 10pt]
		\item We formalize the problem of \ob sampling in the Linear Canonical Transform (LCT) domain, which parametrically generalizes the concept of bandwidth to transforms, including Fresnel, Laplace, Gauss-Weierstrass, and Bargmann transforms (see Table~\ref{tab:trans-summary}). Our theory is fully compatible with the Fourier domain.
		\item We propose 1st and 2nd-order \SDQ architectures for the LCT domain, denoted as \LSDQ. We analyze the boundedness of quantization noise and demonstrate that the recovery error from \ob sampling remains bounded.  
		\item Addressing the saturation problem, we introduce \MLSDQ, or USF-based \LSDQ. We leverage the insights from noise shaping presented in \figref{fig:noiseshaping}(d) which shows that the bandlimited input, modulo folds and quantization noise are segregated in transform domain. This leads to a novel recovery method and introduces the first class of algorithms designed for transform-domain recovery.
		\item We show that the reduction of \MLSDQ to the Fourier domain case achieves superior performance compared to existing time-domain techniques \cite{Graf:2019:C}, offering reduced oversampling requirements and enhanced robustness.
	\end{enumerate}
    \bpara{Notations.} The set of integers, positive integers, real and complex numbers is denoted as $\Z$, $\Zp$, $\R$ and $\C$. The real part of a complex number $x$ is denoted as $\Re\rob{x}$ and the imaginary part of $x$ is denoted as $\Im\rob{x}$. Difference filter is defined as $v\n = \delta\n - \delta\sqb{n - 1}$, and applying this filter to a sequence $a\n$ is denoted as $\diff{a}\n = \l a \conv v \r\n$, $\diffop$ denotes difference operator $(\diffop a)\n = \diff{a}\n$, $\csgn\rob{a\n} = \sgn\rob{\Re\rob{a\n}} + \jmath \sgn\rob{\Im\rob{a\n}}$, $\infnorm{g \n}$ denotes absolute max-norm of $\gn$, $\lfloor \cdot \rfloor$ denotes the floor operator, and $\left\langle {f,g} \right\rangle  = \int {f{g^*}}$ is the $L_2$ inner-product. The mean\nobreakdash-square error (MSE) between $\nvec{x},\nvec{y}\in\C^{N}$ is denoted as $\MSE{\nvec{x}}{\nvec{y}} \DE \frac{1}{N}\sum\nolimits_{n=0}^{N-1} \left|x\sqb{n} -y\sqb{n} \right|^{2}$. $\Unif(a, b)$ is a uniform distribution from $a$ to $b$. Convolution is defined as $\conv$. $\ind_\D(\cdot)$ is the indicator function on domain $\D$, $\rob{\cdot}^H$ denotes the conjugate transpose.
	
	\section{Linear Canonical Transform}
	Linear Canonical Transform (\LCT) was introduced in \cite{Moshinsky:1971:J} in the context of harmonic oscillators in quantum mechanics. The \LCT generalizes a wide class of transformations, which are summarized in \tabref{tab:trans-summary}. Subsequent research on the \LCT includes optimal filtering \cite{Barshan:1997:J}, nonuniform sampling \cite{Sharma:2009:J}, multi-channel consistent sampling \cite{Xu:2017:J}, and sampling of signals bandlimited to a disc \cite{Zayed:2018:J}. The \LCT has also been used for modelling optical cavities \cite{Dahlen:2017:J}, implementing an optical crypto-system \cite{Mohammed:2024:J}, and the optical implementation of 2D \LCT was considered in \cite{Sahin:1998:J}. The \LCT, also referred to as the Affine Fourier Transform (AFT), has shown strong potential in handling doubly-dispersive channels in high-mobility communication environments \cite{Bemani:2023:J, Luo:2024:J, Rou:2024:J}. In what follows, we will revisit some of the definitions associated with the LCT.
	\begin{definition}[Linear Canonical Transform (\LCT)]
		Let $\Lambda = \bigl[ \begin{smallmatrix}
			a & b\\
			c & d\\
		\end{smallmatrix} \bigr]$ with $ad - bc = 1$. The LCT of a function $g\t, t \in \R$, is a unitary, integral mapping: $\L : g \to \LCTop{g}$ defined by
		\begin{equation}
			\L\sqb{g}\w = \LCTop{g}\w = \begin{cases}
				\left \langle g, \kL\rob{\cdot, \omega} \right \rangle & b \neq 0\\
				\sqrt{d} e^{\jmath\frac{1}{2}d\omega^2} & b = 0
			\end{cases}	 . \end{equation}
		The \LCT kernel $\kL \rob{\mat{r}}$, which depends on the time-frequency coordinate $\mat{r} = \sqb{t \ \ \omega}^\top$, is defined as
		\begin{equation}
			\kL \rob{\mat{r}} = \frac{\exp\rob{-\jmath \mat{r}^\top \mat{U} \mat{r}}}{\sqrt{-\jmath 2 \pi b}} , \ \ \mat{U} = \frac{1}{2b} \begin{bmatrix}
				a & -1\\
				-1 & d
			\end{bmatrix}.
		\end{equation}
		The inverse-\LCT is another \LCT defined via $\Linv : \LCTop{g} \to g$,
		\begin{equation}
			g\t = \Linv \sqb{\LCTop{g}}\t = \begin{cases}
				\left \langle \LCTop{g}, \kLinv \rob{\cdot, t} \right \rangle & b \neq 0\\
				\sqrt{a}e^{-\jmath \frac{1}{2} c a t ^ 2} g\rob{at} & b = 0
			\end{cases}.
		\end{equation}
	\end{definition}
	
	When working with sampled sequences, the discrete-time version of the \LCT is a relevant mathematical tool.
	\begin{definition}[Discrete-Time \LCT (\DTLCT)] Let $g\t$ be a continuous-time function with samples $g\n = g\rob{nT}, T > 0$. The \DTLCT of $g[n]$, denoted by $\LCTop{G}\w$, is defined as
		\begin{align*}
			\LCTop{G}\w &= T \left \langle g, \kL\rob{\cdot, \omega} \right \rangle\\
			&= \frac{T}{\sqrt{\jmath 2\pi b}} \sum_{n = -\infty}^{+\infty} g\n e^{\frac{\jmath}{2b} \l a \l n T \r^2 - 2 n T \omega + d \omega^2 \r}.
		\end{align*}
		The Inverse Discrete-Time \LCT (\IDTLCT) is then defined as
		\begin{align*}
			\gn &= \left \langle \LCTop{G}, \kLinv\rob{\cdot, nT} \right \rangle_{\sqb{-\pi b/T,\pi b/T}}\\
			&= \frac{1}{\sqrt{-\jmath 2\pi b}} \int_{-\pi b / T}^{\pi b / T} \LCTop{G}\w e^{-\frac{\jmath}{2b} \l a\l n T\r^2 - 2 n T \omega + d \omega^2 \r} d\omega.
		\end{align*}
	\end{definition}
	When working with compactly supported functions, Linear Canonical Series (\LCS) parallels the Fourier Series.
	\begin{definition}[Linear Canonical Series (\LCS)]
		Let $g\t$ be a continuous-time function on the interval $t \in [0, \tau)$. Its LCS is
		\begin{equation}
			g\t = \kappa_{b,\tau} \sum\nolimits_{m\in\Z} \LCTop{g}\sqb{m} \kL \rob{t, m \omega_0 b}, \omega_0 = \frac{2 \pi}{\tau}
		\end{equation}
		where $\kappa_{b,\tau} = \sqrt{\omega_0 b}$. The LCS coefficients are given by $\LCTop{g}\sqb{m} = \kappa_{b,\tau}\L\sqb{g}\rob{m \omega_0 b} \equiv \kappa_{b,\tau}\left \langle g, \kL\rob{\cdot, m \omega_0 b} \right \rangle$.
	\end{definition}
	Similar to the Discrete Fourier Transform (\DFT), we define the Discrete \LCT (\DLCT) as follows.
	\begin{definition}[Discrete \LCT (\DLCT)]
		Let $g\t$ be defined on the interval $[0, \tau)$ with samples $g\n = g\rob{nT}, T > 0, n \in [0, N - 1]$ and $N = \left \lfloor \frac{\tau}{T} \right \rfloor$, its \DLCT is defined as
		\begin{align}
			\LCTop{g}\m = T\kappa_{b,\tau}\sum\nolimits_{n=0}^{N-1} g\n \kLinv \rob{m \omega_0 b, nT}
			\label{eq:dlct-def}
		\end{align}
		where $\omega_0 = \frac{2\pi}{\tau}$. The Inverse \DLCT (\IDLCT) is defined as
		\begin{equation}
			g\n = \kappa_{b,\tau}\sum\nolimits_{\abs{m} \leq M} \LCTop{g}\m \kL \rob{nT, m \omega_0 b}.
		\end{equation}
	\end{definition}
	Since basis functions for \LCT are chirps and not complex exponentials, it is helpful to define chirp (de)modulation.
	\begin{definition}[Chirp (De)modulation]
		Let $\mat{\Lambda} = \bigl[ \begin{smallmatrix} a & b \\ c & d \end{smallmatrix} \bigr]$ be  $2 \times 2$ matrix. The chirp modulation function is defined as
		\begin{equation}
			m_{\mat{\Lambda}}\t \DE \exp \rob{\jmath \frac{a}{2b} t ^ 2}.
			\label{eq:chirp-def}
		\end{equation}
		${\mat{\Lambda}}$-parametrized up and down chirp modulation is denoted by
		\begin{equation}
			\uparr{g}\t \DE m_{\mat{\Lambda}}\t g \t \quad\text{and}\quad\dnarr{g}\t \DE m_{\mat{\Lambda}}^*\t g\t.
		\end{equation}
	\end{definition}
	Standard convolution $\conv$ does not yield multiplication in the \LCT domain. Hence, we define the \LCT convolution $\LCTconv$.
	\begin{definition}[\LCT Convolution and Product Theorem \cite{Bhandari:2019:J}]
		\label{def:lct-conv-prod-theorem}
		Let $\{g,f\}$ be continuous-time functions with samples $\{\gn,f\sqb{n}\} $, respectively. The \LCT convolution is defined as
		\begin{equation}
			h\n = \l f \LCTconv g \r\n \DE \KL \dnL \n \l \uparr{f} \conv \uparr{g} \r \n
			\label{eq:lct-conv-time}
		\end{equation}
		where $\KL = \frac{1}{\sqrt{\jmath 2 \pi b}}$. The DT-\LCT of $h\n$ is
		\begin{equation}
			h\n \xrightarrow{\mathsf{LCT}} {\LCTop{H}\w =  \frac{\Phi_\Lmtx\w}{T} \LCTop{F}\w \LCTop{G} \w}
		\end{equation}
		where $\{\LCTop{F},\LCTop{G},\LCTop{H} \}$ represent the DT-LCT of sequences, $\{f\n,g\n,h\n \}$, respectively and $\Phi_\Lmtx\w = e^{-\jmath\frac{d \omega^2}{2b}}$.
	\end{definition}
	For signals that are bandlimited in the \LCT domain, the following extension of Shannon's sampling theorem guarantees recovery from samples  \cite{Bhandari:2019:J}. 
	\begin{theorem} Let $g\in\BOmLCT$, then, provided that $T\leq {\pi b}/{\Omega_m}$, $g\t$ can be recovered from  samples via
		\begin{equation}
			g\t = e^{-\jmath\tfrac{a t^2}{2b}}\sum\limits_{n\in\Z} \uparr{g}\rob{nT} \sinc \rob{\frac{t - nT}{T}}.
		\end{equation}
	\end{theorem}
	
	\begin{table}[t]
		\centering
		\caption{Summary of Transformations}
		\begin{tabular}{p{3.2cm}  p{4.8cm}}
			\hline
			\rowcolor{blue!20} 
			{ \LCT Parameters} $\left(\bsls \right) $ & { Corresponding Transformations} \\
			\hline
			\addlinespace
			
			$\bigl[ \begin{smallmatrix} & 0 & 1 \\- & 1 &0  \end{smallmatrix} \bigr] 
			= \pmb\Lambda_\textsf{FT}$ &  \textbf{Fourier Transform (\FT)}   \\ 		[4.5pt] 
			
			$\bigl[ \begin{smallmatrix} &\cos\theta&\sin\theta \\
				-&\sin\theta&\cos\theta \end{smallmatrix} \bigr] = \pmb\Lambda_\theta $	 &  \textbf{Fractional Fourier Transform (\FrFT)} \\	[4.5pt] 
			
			$\bigl[ \begin{smallmatrix} & 1 & b \\ & 0 & 1 \end{smallmatrix} \bigr] = \pmb\Lambda_\textsf{FrT}$ & \textbf{Fresnel Transform (\Fr)}  \\	[4.5pt] 
			
			$\bigl[ \begin{smallmatrix} & 0 & \j \\ & \j & 0 \end{smallmatrix} \bigr] 
			= \pmb\Lambda_\textsf{LT}$ & \textbf{Laplace Transform (LT)}  \\	[4.5pt] 
			
			$\bigl[ \begin{smallmatrix} & \j \cos\theta &\j \sin\theta  \\
				& \j \sin \theta & -\j\cos\theta \end{smallmatrix} \bigr] $ & \textbf{Fractional Laplace Transform}   \\ 	[4.5pt] 
			
			$\bigl[ \begin{smallmatrix} & 1 &\jmath b \\  & \jmath & 1 \end{smallmatrix} \bigr]$ & \textbf{Bilateral Laplace Transform}  \\	[4.5pt] 
			
			$\bigl[ \begin{smallmatrix} & 1 & -\jmath b  \\ & 0 & 1 \end{smallmatrix} \bigr]$, $b\;\geqslant\;0$  & \textbf{Gauss--Weierstrass Transform}   \\	[4.5pt] 
			
			$\tfrac{1}{{\sqrt 2 }} \bigl[\begin{smallmatrix} & 0 & e^{ - {{\jmath\pi } 
						\mathord{\left/{\vphantom {{j\pi } 2}} \right.\kern-\nulldelimiterspace} 2}} \\
				& -e^{ - {{\jmath\pi } \mathord{\left/{\vphantom {{j\pi } 2}} \right.\kern-\nulldelimiterspace} 2}} 
				& 1 \end{smallmatrix} \bigr]$ & \textbf{Bargmann Transform} \\	[4pt] 
			
			\addlinespace
			\hline
		\end{tabular}
		\label{tab:trans-summary}
	\end{table}
	
	\section {Revisiting 1-Bit Sampling in Fourier Domain}
	Here, we review the fundamental principles of \SDQ for Fourier-bandlimited signals. This will be the stepping stone towards the extension of \SDQ to a broader class of transformations (as outlined in Table~\ref{tab:trans-summary}) in the next section.
	We denote an $\Omega$-bandlimited function in the Fourier domain by,
	\begin{equation}
		\label{eq:FBL}
		g \in \BOmFT \leftrightarrow \LCTOpFT{g}\w = \LCTOpFT{g}\w \ind_{\sqb{-\OmFT, \OmFT}}\w.
	\end{equation}
	Such signals can be reconstructed from $\{\pm 1\}$ samples obtained via \SDQ, with a bounded reconstruction error \cite{Daubechies:2003:J}. \SDQ is a low-complexity acquisition scheme that employs a feedback loop with a signum function. For bounded signals, say ${|g|\leq 1}$, the \SDQ scheme is summarized as follows \cite{Daubechies:2003:J}:
	\begin{subequations}
		\label{eq:sdq-ft}
		\begin{empheq}[box=\shadowbox*]{align}
			\un &= u\sqb{n - 1} + g \rob{n\Tcur} - \qn,  \quad u \in \l -1, 1\r 
			\label{eq:un-ft} \\
			\qn &= \sgn \rob{u\sqb{n - 1} + g \rob{n\Tcur}}, \qquad T = \pi/\Omega
		\end{empheq}	
	\end{subequations}
	where $u$ is an intermediate variable and $\osrat \in \l 0, 1\r$ denotes the oversampling ratio. We rearrange \eqref{eq:un-ft} to establish the relation between \ob samples $q$ and the bounded input signal $g$
	\begin{equation}
		\qn = \underbrace{g \rob{n\Tcur}}_{g \in \BOmFT} - \underbrace{\l u \conv \diffflt \r \n}_\text{High-Pass} \equiv 
		{g \rob{n\Tcur}} - \rob{\Delta u}\sqb{n}.
		\label{eq:qn-noise-shaping}
	\end{equation}
	Since $\diffflt\n$ is a high-pass filter, it impels the quantization noise into high frequencies, moving it away from $\LCTOpFT{g}\w$, aptly justifying why \SDQ is associated with \nosh \cite{Daubechies:2003:J}. The net effect is that one can recover $\gt$ by filtering $\qn$,
	\begin{equation}
		\tilde{g}\t = \osrat \sum\limits_{n \in \Z} \qn \idker_\Omega\rob{\frac{t}{T} - n\osrat}
	\end{equation}
	where $\idker_\Omega$ is a low-pass interpolation function \cite{Daubechies:2003:J}. Not surprisingly, in line with conventional methods, redundancy improves the reconstruction quality as noise shaping is more effective---$\un$ is displaced further away from the low-pass interval of $\LCTOpFT{g}\w$. The reconstruction error can be expressed as
	\begin{align*}
		e\t \DE \rob{g-\tilde{g}}\rob{t} =  \osrat \underbrace{\sum\nolimits_{n \in \Z} \l u \conv v \r \n \idker \rob{\tfrac{t}{T} - n\osrat}}_\text{Low-pass filtered quantization noise}.
	\end{align*}
	The following result quantifies the relationship between oversampling ratio $\rob{\osrat}$ and reconstruction quality:
	\begin{equation}
		\abs{e\t} \leq \osrat \norm{\partial_t \idker_\Omega}_{L^1}
		\label{eq:ft-sdq-err-bound}
	\end{equation}
	which can be further improved to $\O\rob{\osrat^3}$ (see \cite{Daubechies:2003:J} Section 2.3).
	
	\section{Towards 1-Bit Sampling in \LCT Domain}
	\label{sec:1BLCT}
	Here, we develop the \LSDQ scheme---the \SDQ scheme for the LCT domain. Clearly, the \LSDQ scheme should maintain backwards compatibility with the conventional Fourier domain \SDQ. The first step towards this goal is to generalize the notation of bandlimitedness in \eqref{eq:FBL}. We do so by defining the class of $\Omega$-bandlimited signals in the \LCT domain as,
	\begin{equation}
		g \in \BOmLCT \leftrightarrow \LCTop{g}\w = \LCTop{g}\w \ind_{[-\OmL, \OmL]} \w.
	\end{equation}
	This generalizes the notion of bandlimitedness to a much wider class of transformations (see \tabref{tab:trans-summary}). Such signals, better described by their \LCT, have been studied in optics \cite{Souvorov:2006:J, Wolf:2007:J}, harmonic analysis \cite{Zayed:2021:J}, signal processing \cite{Liebling:2003:J, Bhandari:2012:J}, or communications \cite{Luo:2024:J, Rou:2024:J}. Despite these research efforts, \SDQ for \LCT-bandlimited signals was not studied previously.
	
	Clearly, whenever $\Lmtx \neq \LmtxFT$, the difficulty is that \nosh can no longer be achieved via \eqref{eq:qn-noise-shaping} because of the breakdown of the well-known convolution-multiplication property\footnote{That is, when working with LCTs, convolution in one domain does not imply multiplication in another domain. } of the Fourier transforms for the LCT domain. This necessitates the development of a new strategy. Here, we will develop 1st and 2nd-order \LSDQ schemes. Higher-order generalizations may be considered in future works.
	
	\bpara{1st-order \LSDQ.} Since our goal is to achieve noise shaping for any $g \in \BOmLCT$, it is natural to consider the LCT of the form,
	\begin{equation}
		\LCTop{Q}\w = \underbrace{\LCTop{G}\w}_\text{$g \in \BOmLCT$} - \underbrace{\LCTop{U}\w \l 1 - e^{-\jmath \frac{\omega \Tcur}{b}}\r}_\text{High-Pass}
		\label{eq:q-lct-def}
	\end{equation}
	where $\LCTop{Q}\w$ denotes the \DTLCT (see \defref{def:lct-conv-prod-theorem}) of \ob samples, $\qLn$. In the time domain, one would expect
	\begin{equation}
		\qLn = {\gn}
		- {\l u \LCTconv \vL \r \n}
		\label{eq:qLn-hp}
	\end{equation}
	where $\vL$ (noise shaping filter) and $u$ remain to be characterized. Note that we have utilized the LCT convolution operator, \ie, $\LCTconv$, which respects the convolution and product theorem \cite{Bhandari:2012:J} in the LCT domain, and allows us to represent \eqref{eq:qLn-hp} as,
	\begin{equation}
		\LCTop{Q}\w = \LCTop{G}\w - \frac{e^{-\jmath \frac{d \omega ^ 2}{2b}}}{\Tcur} \LCTop{U}\w \LCTop{V} \w.
		\label{eq:qw-lct-conv}
	\end{equation}
	We can now relate \eqref{eq:qw-lct-conv} with \eqref{eq:q-lct-def}, which leads to the identification of the noise shaping filter in terms of LCT parameters,
	\begin{align}
		\label{eq:Vw}
		\LCTop{V} \w &= \rob{\Tcur} e^{\jmath\frac{d\omega^2}{2b}} \l 1 - e^{-\jmath \frac{\omega \Tcur}{b}}\r, \quad T = {\pi b}/{\Omega_m}      
	\end{align}
	where the magnitude response implements a high-pass filter in the LCT domain, or $|{\LCTop{V}\w}| = 2 \Tcur \abs{\sin \l {\omega \Tcur / \rob{2 b}} \r}$.
	
	To express the state equations of \LSDQ, we need to deduce $\un$ from \eqref{eq:qLn-hp}. To this end, we first identify $\vLn$ as follows,
	\begin{align*}
		\vLn & = \left \langle \LCTop{V}, \kLinv \rob{\cdot, n \Tcur} \right \rangle \\
		&= \frac{e^{-\jmath\frac{a\rob{n\Tcur}^2}{2b}}}{\sqrt{-\jmath 2 \pi b}}  
		\rob{\Tcur}
		\int_{-\pi b / \rob{\Tcur}}^{\pi b / \rob{\Tcur} }  {\left( {1 - {e^{ - \jmath \frac{{\omega hT}}{b}}}} \right)} {e^{\jmath \frac{{\omega hT}}{b}n}}d\omega\\ 
		&= \sqrt{\jmath 2 \pi b} \dnarr{\diffflt} \n.
	\end{align*}
	Setting $\vLn = \sqrt{\jmath 2 \pi b} \dnarr{\diffflt}$, next we re-write \eqref{eq:qLn-hp} as follows,
	\begin{align}
		\qLn &= \gn - \KL \dnL\n \l \uparr{u} \conv \uparr{v}_\Lmtx \r \n \notag\\
		&= \rob{g - u}\n + \dnL\n \upL\sqb{n - 1} u\sqb{n - 1}.
		\label{eq:qL-exp}
	\end{align}
	Now, since $\dnL\n \upL\sqb{n - 1} = e^{-\jmath\frac{a\rob{2n - 1} \l \Tcur \r^2}{2b} }$ in the above, we can express the intermediate state variable $\un$ as
	\begin{equation}
		u\n = \gn -  \qLn + e^{-\jmath\frac{a\rob{2n - 1} \l \Tcur \r^2}{2b} } u\sqb{n - 1}.
		\label{eq:un-sdq-lct}
	\end{equation}
	In analogy to \eqref{eq:sdq-ft}, the 1st-order \LSDQ is written as follows,	
	\begin{subequations}
		\begin{empheq}[box=\shadowbox*]{align}
			\un &= \gn -  \qLn + e^{-\jmath\frac{a\rob{2n - 1}\l \Tcur \r^2}{2b}} u\sqb{n - 1}\label{eq:un-lct-proof}\\
			\qLn &= \csgn \rob{e^{-\jmath\frac{a\rob{2n - 1}\l \Tcur \r^2}{2b}} u \sqb{n - 1} + \gn}\label{eq:sdq-lct-proof}.
		\end{empheq}
		\label{eq:lsdq-1}
	\end{subequations}	
    The novel \LSDQ architecture is depicted in \figref{fig:lctsdq}. We reconstruct $g\t$ by low-pass filtering \ob samples $\qLn$ with interpolation kernel $\idker$ with bandwidth $\Omega$ as follows,
	\begin{align}
		\grec\t & = \osrat e^{-\jmath \frac{a t^2}{2b}} \sum\limits_{n \in \Z} \uparr{q}_{\Lmtx}\n \idker_\Omega \rob{\frac{t}{T} - n \osrat}\notag\\
		&\EQc{eq:qLn-hp} \gt - \underbrace{\osrat e^{-\jmath \frac{a t^2}{2b}} \sum\limits_{n \in \Z} \l \uparr{u} \conv \diffflt \r \n \idker \rob{\frac{t}{T} - n \osrat}}_{e\t}\label{eq:lsdq1-rec}
	\end{align}
	where the approximation error $e\t$ is attributed to 
	$u \LCTconv \vL$---the contamination of the baseband by the quantization noise. Since the approximation error is defined as $e\t = g\t - \tilde{g}\t$, we are interested in bounding its value. Before bounding the approximation error, we first bound the state variable $\un$.
	\begin{figure}[!t]
		\centering
		\scalebox{0.5}{
			\tikzset{every picture/.style={line width=0.75pt}} 
			\begin{tikzpicture}[x=0.75pt,y=0.75pt,yscale=-1,xscale=1]
				\draw  [line width=1.5]  (111,129.75) .. controls (111,125.33) and (114.58,121.75) .. (119,121.75) .. controls (123.42,121.75) and (127,125.33) .. (127,129.75) .. controls (127,134.17) and (123.42,137.75) .. (119,137.75) .. controls (114.58,137.75) and (111,134.17) .. (111,129.75) -- cycle ; \draw  [line width=1.5]  (111,129.75) -- (127,129.75) ; \draw  [line width=1.5]  (119,121.75) -- (119,137.75) ;
				\draw  [line width=1.5]  (157.5,129.75) .. controls (157.5,125.33) and (161.08,121.75) .. (165.5,121.75) .. controls (169.92,121.75) and (173.5,125.33) .. (173.5,129.75) .. controls (173.5,134.17) and (169.92,137.75) .. (165.5,137.75) .. controls (161.08,137.75) and (157.5,134.17) .. (157.5,129.75) -- cycle ; \draw  [line width=1.5]  (157.5,129.75) -- (173.5,129.75) ; \draw  [line width=1.5]  (165.5,121.75) -- (165.5,137.75) ;
				\draw  [line width=1.5]  (214.5,109) -- (259.5,109) -- (259.5,150.5) -- (214.5,150.5) -- cycle ;
				\draw [line width=1.5]    (65,129.75) -- (107,129.75) ;
				\draw [shift={(111,129.75)}, rotate = 180] [fill={rgb, 255:red, 0; green, 0; blue, 0 }  ][line width=0.08]  [draw opacity=0] (13.4,-6.43) -- (0,0) -- (13.4,6.44) -- (8.9,0) -- cycle    ;
				\draw [line width=1.5]    (174.5,129.75) -- (211.5,129.75) ;
				\draw [shift={(215.5,129.75)}, rotate = 180] [fill={rgb, 255:red, 0; green, 0; blue, 0 }  ][line width=0.08]  [draw opacity=0] (13.4,-6.43) -- (0,0) -- (13.4,6.44) -- (8.9,0) -- cycle    ;
				\draw  [line width=1.5]  (345.5,129.75) .. controls (345.5,125.07) and (349.29,121.28) .. (353.97,121.28) .. controls (358.65,121.28) and (362.44,125.07) .. (362.44,129.75) .. controls (362.44,134.43) and (358.65,138.22) .. (353.97,138.22) .. controls (349.29,138.22) and (345.5,134.43) .. (345.5,129.75) -- cycle ; \draw  [line width=1.5]  (347.98,123.76) -- (359.96,135.74) ; \draw  [line width=1.5]  (359.96,123.76) -- (347.98,135.74) ;
				\draw  [line width=1.5]  (437.33,129.75) .. controls (437.33,125.33) and (440.92,121.75) .. (445.33,121.75) .. controls (449.75,121.75) and (453.33,125.33) .. (453.33,129.75) .. controls (453.33,134.17) and (449.75,137.75) .. (445.33,137.75) .. controls (440.92,137.75) and (437.33,134.17) .. (437.33,129.75) -- cycle ; \draw  [line width=1.5]  (437.33,129.75) -- (453.33,129.75) ; \draw  [line width=1.5]  (445.33,121.75) -- (445.33,137.75) ;
				\draw [line width=1.5]    (363,129.75) -- (433.33,129.75) ;
				\draw [shift={(437.33,129.75)}, rotate = 180] [fill={rgb, 255:red, 0; green, 0; blue, 0 }  ][line width=0.08]  [draw opacity=0] (13.4,-6.43) -- (0,0) -- (13.4,6.44) -- (8.9,0) -- cycle    ;
				\draw [line width=1.5]    (453.33,129.75) -- (490.33,129.75) ;
				\draw [shift={(494.33,129.75)}, rotate = 180] [fill={rgb, 255:red, 0; green, 0; blue, 0 }  ][line width=0.08]  [draw opacity=0] (13.4,-6.43) -- (0,0) -- (13.4,6.44) -- (8.9,0) -- cycle    ;
				\draw [line width=1.5]    (259.5,129.75) -- (341.5,129.75) ;
				\draw [shift={(345.5,129.75)}, rotate = 180] [fill={rgb, 255:red, 0; green, 0; blue, 0 }  ][line width=0.08]  [draw opacity=0] (13.4,-6.43) -- (0,0) -- (13.4,6.44) -- (8.9,0) -- cycle    ;
				\draw [line width=1.5]    (353.97,138.22) -- (353.97,173) -- (165.5,173) -- (165.5,141.75) ;
				\draw [shift={(165.5,137.75)}, rotate = 90] [fill={rgb, 255:red, 0; green, 0; blue, 0 }  ][line width=0.08]  [draw opacity=0] (13.4,-6.43) -- (0,0) -- (13.4,6.44) -- (8.9,0) -- cycle    ;
				\draw [line width=1.5]    (353.97,77.67) -- (353.97,117.28) ;
				\draw [shift={(353.97,121.28)}, rotate = 270] [fill={rgb, 255:red, 0; green, 0; blue, 0 }  ][line width=0.08]  [draw opacity=0] (13.4,-6.43) -- (0,0) -- (13.4,6.44) -- (8.9,0) -- cycle    ;
				\draw [line width=1.5]    (127,129.75) -- (153.5,129.75) ;
				\draw [shift={(157.5,129.75)}, rotate = 180] [fill={rgb, 255:red, 0; green, 0; blue, 0 }  ][line width=0.08]  [draw opacity=0] (13.4,-6.43) -- (0,0) -- (13.4,6.44) -- (8.9,0) -- cycle    ;
				\draw  [line width=1.5]  (494.5,108.5) -- (575,108.5) -- (575,150.36) -- (494.5,150.36) -- cycle ;
				\draw [line width=1.5]    (598.33,129.75) -- (598.33,207.4) -- (119,208.5) -- (119,141.75) ;
				\draw [shift={(119,137.75)}, rotate = 90] [fill={rgb, 255:red, 0; green, 0; blue, 0 }  ][line width=0.08]  [draw opacity=0] (13.4,-6.43) -- (0,0) -- (13.4,6.44) -- (8.9,0) -- cycle    ;
				\draw [shift={(598.33,129.75)}, rotate = 90] [color={rgb, 255:red, 0; green, 0; blue, 0 }  ][fill={rgb, 255:red, 0; green, 0; blue, 0 }  ][line width=1.5]      (0, 0) circle [x radius= 4.36, y radius= 4.36]   ;
				\draw [line width=1.5]    (575.33,129.75) -- (627.4,129.75) ;
				\draw [shift={(631.4,129.75)}, rotate = 180] [fill={rgb, 255:red, 0; green, 0; blue, 0 }  ][line width=0.08]  [draw opacity=0] (13.4,-6.43) -- (0,0) -- (13.4,6.44) -- (8.9,0) -- cycle    ;
				\draw [line width=1.5]    (84,129.75) -- (84,18.33) -- (445,18.33) -- (445.32,117.75) ;
				\draw [shift={(445.33,121.75)}, rotate = 269.82] [fill={rgb, 255:red, 0; green, 0; blue, 0 }  ][line width=0.08]  [draw opacity=0] (13.4,-6.43) -- (0,0) -- (13.4,6.44) -- (8.9,0) -- cycle    ;
				\draw [shift={(84,129.75)}, rotate = 270] [color={rgb, 255:red, 0; green, 0; blue, 0 }  ][fill={rgb, 255:red, 0; green, 0; blue, 0 }  ][line width=1.5]      (0, 0) circle [x radius= 4.36, y radius= 4.36]   ;
				\draw [color={rgb, 255:red, 0; green, 0; blue, 0 }  ,draw opacity=1 ][line width=2.25]    (99,158.2) -- (110,158.2) ;
				
				\draw (39.92,129.56) node  [font=\LARGE]  {$g\n$};
				\draw (237,129.75) node  [font=\LARGE]  {$z^{-1}$};
				\draw (350.92,58.4) node  [font=\Large]  {$\exp \l -\jmath \frac{a\rob{2n-1}\rob{\Tcur}^{2}}{2b} \r$};
				\draw (190.42,103.9) node  [font=\LARGE]  {$u\n$};
				\draw (305.42,103.9) node  [font=\LARGE]  {$u\sqb{n - 1}$};
				\draw (662.12,130.06) node  [font=\LARGE]  {$\qLn$};
				\draw (534.75,129.43) node  [font=\LARGE]  {$\csgn\rob{ \cdot }$};

			\end{tikzpicture}
		}
		\caption{System architecture for 1st-order \LSDQ described by \eqref{eq:lsdq-1}.}
		\label{fig:lctsdq}
	\end{figure}
	
	\begin{lemma}[Boundedness Property]
		\label{lem:BP}
		Assume $\abs{g\t} \leqslant 1$, $\abs{\Re\rob{u\sqb{0}}} < 1$ and $\abs{\Im\rob{u\sqb{0}}} < 1$, then for $\{\un\}_{n \in \Zp}$ defined in \eqref{eq:un-lct-proof}, it holds that $\abs{\Re\rob{\un}} < 1$ and $\abs{\Im\rob{\un}} < 1$.
	\end{lemma}
	\begin{proof}
		Let $p_\Lmtx\n = e^{-\jmath \frac{a\rob{2n - 1}\rob{\Tcur}^2}{2b}}$, clearly $\abs{p_\Lmtx} = 1$. By induction, $\abs{\Re\rob{u\sqb{n - 1}}} < 1$ and $\abs{\Im\rob{u\sqb{n - 1}}} < 1$, then, $\abs{\Re\rob{p_\Lmtx\n u\sqb{n - 1}}} < 1$ and $\abs{\Im\rob{p_\Lmtx\n u\sqb{n - 1}}} < 1$. Since $\infnorm{g} \leqslant 1$, then $\abs{\Re\rob{p_\Lmtx\n u \sqb{n - 1} + \gn}} < 2$ and $\abs{\Im\rob{p_\Lmtx\n u \sqb{n - 1} + \gn}} < 2$. Assuming $\qLn$ defined in \eqref{eq:sdq-lct-proof}, $\abs{\Re{\rob{\qLn}}} = 1$ and $\abs{\Im{\rob{\qLn}}} = 1$ which implies that $\abs{\Re\rob{\un}} = \abs{\Re\rob{p_\Lmtx\n u \sqb{n -1} + g\n - \qLn}} < 1$ and $\abs{\Im\rob{\un}} = \abs{\Im\rob{p_\Lmtx\n u \sqb{n -1} + g\n - \qLn}} < 1$.
	\end{proof}
	With $\un$ bounded, next we show that the approximation error, $e\rob{t} =g\rob{t} - \widetilde{g}\rob{t}$ (see \eqref{eq:lsdq1-rec}), is also bounded.
	\begin{proposition}[Recovery Error Bound]
		\label{pr:err-bound}
		Let $g \in \BOmLCT, \infnorm{g} \leqslant 1$ with \ob samples, $\qLn$ \eqref{eq:sdq-lct-proof}. Then the following holds,
		\begin{align}
			\label{eq:etbound}
			e\rob{t} & = \osrat e^{-\jmath \frac{a t^2}{2b}} \sum\limits_{n \in \Z} \l \uparr{u} \conv \diffflt \r \n \idker \rob{\frac{t}{T} - n \osrat}, \ \ \  \osrat < \l 0, 1 \r \notag \\ 
			& \Longrightarrow \abs{e\t} \leqslant \osrat \sqrt{2} \norm{\partial_t \idker_\Omega}_{L^1}.
		\end{align}
	\end{proposition}
	
	\begin{proof}
		With $\abs{e\t} = \osrat \left |  \sum\nolimits_{n\in\Z} \l \uparr{u} \conv v \r \n \idker_{\Omega} \l \frac{t}{T} - n\osrat \r \right |$ and since $|{\uparr{u}\n}| = \sqrt{2}$, we can simplify $\abs{e\t}$ as follows,
		\begin{align*}
			&\abs{e\t} = \osrat \left |  \sum\nolimits_{m\in\Z} \uparr{u}\m  \sum\nolimits_{n \in \Z} v\sqb{n - m} \idker_{\Omega} \l \tfrac{t}{T} - n\osrat \r \right | \\
			&\leq \osrat \sqrt{2} \sum\nolimits_{n \in \Z} \left | \l \idker_{\Omega} \l \tfrac{t}{T} - n\osrat \r -  \idker_{\Omega} \l \tfrac{t}{T} - \rob{n + 1}\osrat \r \r \right |\\
			& \leq \osrat \sqrt{2} \sum\nolimits_{n \in \Z} \int_{\frac{t}{T} - \rob{n + 1}\osrat}^{\frac{t}{T} - n\osrat}\left | \partial_t \idker \rob{y}  dy \right | = \osrat \sqrt{2} \norm{\partial_t \idker_\Omega}_{L^1}.		
		\end{align*}
		This concludes the proof. The additional $\sqrt{2}$ factor is due to the use of complex numbers. For the FT case and real input, replacing $\csgn$ in \eqref{eq:sdq-lct-proof} with $\sgn$ leads to the conventional $\SD$ architecture in \eqref{eq:sdq-ft} with error bound shown in \eqref{eq:ft-sdq-err-bound}.
	\end{proof}
	
	\bpara{2nd-order \LSDQ.} Using a 2nd-order difference filter offers improved noise shaping in the Fourier domain \cite{Daubechies:2003:J}. Here, we leverage this property for \LCT and derive the corresponding 2nd-order scheme. In the view of \eqref{eq:q-lct-def}, replacing $\l 1 - e^{-\jmath \rob{\omega \Tcur/b}} \r$ with $\l 1 - e^{-\jmath \rob{\omega \Tcur/b}}\r^2$ improves noise rejection. \DTLCT of 2nd-order \LSDQ, akin to \eqref{eq:q-lct-def}, satisfies
	\begin{align}
		&\LCTop{Q}^{[2]}\w 
		= {\LCTop{G}\w} - {\LCTop{U}\w \l 1 - e^{-\jmath \frac{\omega \Tcur}{b}} \r ^ 2} \label{eq:sd2-lct-req} \\
		&\EQc{eq:Vw} \LCTop{G}\w - 
		\rob{\tfrac{e^{-\jmath \frac{d\omega^2}{2b}}}{\Tcur} }
		\LCTop{U}\w \l 1 - e^{-\jmath \frac{\omega \Tcur}{b}} \r \LCTop{V}\w\notag
	\end{align}
	where $\LCTop{Q}^{[2]}\w$ denotes the \DTLCT of $\qLsecn$. Since $\qLsecn$ utilizes a 2nd-order filter (see \eqref{eq:sd2-lct-req}), we expect an auto-convolution structure in the time domain,
	\begin{equation}
		\qLsecn = \gn - \rob{u \LCTconv \vL \LCTconv \vL}\n.
		\label{eq:qL2-n}
	\end{equation}

	By introducing $x = u \LCTconv \vL$, we can interpret \eqref{eq:qL2-n} as a 1st-order \LSDQ, enabling a simplified representation,
	\begin{equation}
		\qLsecn = \gn - \l x \LCTconv \vL \r \n
		\label{eq:qlsec-n}
	\end{equation}
	with $\LCTop{Q}^{[2]}\w  = \LCTop{G}\w - \LCTop{X}\w \l 1 - \exp\rob{-\jmath \frac{\omega \Tcur}{b}} \r$ being the \DTLCT representation.
    We now express $x\n$ as
	\begin{equation}
		x\n = \gn -  \qLsecn + \underbrace{e^{-\jmath\frac{a\rob{2n - 1} \l \osrat T \r^2}{2b} } x\sqb{n - 1}}_{\tilde{x}\n}.
	\end{equation}
	In the \DTLCT domain, $x\n$ becomes \eqref{eq:sd2-lct-req}, whenever
	\begin{align*}
		\LCTop{X}\w = \frac{e^{-\jmath \frac{d\omega^2}{2b}}}{\Tcur} \LCTop{U}\w \LCTop{V} \w = \LCTop{U}\w \l 1 - e^{-\jmath \frac{\omega \Tcur}{b}} \r.
	\end{align*}
	
	We can now relate $\un$ with $\xn$ in the time domain, 
	\begin{align}
		\xn &= \l u \LCTconv \vL \r \n = \dnL\n \l \uparr{u} \conv v \r \n \notag\\
		&= \un - e^{-\jmath\frac{a\rob{2n - 1} \rob{\Tcur}^2}{2b}} u\sqb{n - 1}.
		\label{eq:xn-lctconv-vL}
	\end{align}
	From \eqref{eq:xn-lctconv-vL} we express $\un$ as
	\begin{equation}
		\un = \xn + {e^{-\jmath\frac{a\rob{2n - 1} \rob{\Tcur}^2}{2b}} u\sqb{n - 1}} \equiv x\n + {\tilde{u}\n}.
	\end{equation}
	Lastly, in analogy to the \FT case \cite{Daubechies:2003:J}, we set $\qLsecn$ as follows
	\begin{equation}
		\qLsecn = \csgn\rob{G\rob{\tilde{x}\n, \tilde{u}\n, \gn}}
	\end{equation}
	where $G\rob{x, u, g} = c_0 x + u + g$ with $c_0 = \sfrac{1}{2}$. Our construction of the 2nd\nobreakdash-order \LSDQ is shown in \figref{fig:lctsdq2}, and it is described by the following state equations:
	\begin{subequations}
		\label{eq:sdq2-lct-lemma}
		\begin{empheq}[box=\shadowbox*]{align}
			\xn &= \gn -  \qLsecn + e^{-\jmath\frac{a\rob{2n - 1} \rob{\Tcur}^2}{2b}} x\sqb{n - 1}\label{eq:xn2-lct-def}\\
			\un &=  \xn + e^{-\jmath\frac{a \rob{2n - 1} \rob{\Tcur}^2}{2b}} u \sqb{n - 1}\label{eq:un2-lct-def}\\
			\qLsecn &= \csgn \rob{c_0 \tilde{x}\n + \tilde{u}\n + \gn},\quad c_0 = \frac{1}{2}.\label{eq:qn2-lct-def}
		\end{empheq}
	\end{subequations}%
	The reconstruction of $g\t$ is obtained by low-pass filtering $\qLsecn$ with any low-pass kernel $\idker_\Omega \in \BOmLCT$
	\begin{equation}
		\grec\t = \osrat e^{-\jmath \frac{a t^2}{2b}} \sum\limits_{n \in \Z} \qLsecn \idker_\Omega \rob{\frac{t}{T} - n \osrat}.
	\end{equation}
	The approximation error $e = g- \tilde{g}$ measures contamination of $\LCTop{g}$ by filtered quantization noise.

	\section{\USF Acquisition}
	For conventional \SDQ to function properly, it is essential that $\infnorm{g\t} < 1$ (see Lemma~\ref{lem:BP}), and this limitation also applies to \LSDQ schemes. Whenever $\infnorm{g\t} > 1$, the reconstruction fails (see Fig. 3 in \cite{Graf:2019:C}). To prevent the quantizer from saturating, we leverage the \USF strategy and introduce modulo non-linearity prior to \LSDQ, defined by,
	\begin{equation}
		\Mmod_{\lambda} \rob{g} = 2\lambda \l \fracprt{\frac{g}{2\lambda} + \frac{1}{2}} - \frac{1}{2} \r, \ \ \
		\begin{array}
			{l}{\lambda\in\R^+}\\{\llbracket g \rrbracket = g - \lfloor g \rfloor}
		\end{array}.
	\end{equation}
	For complex-valued functions, $g\in\C$ we define,
	\begin{equation}
		\Mmod_{\lambda} \rob{g} = \Mmod_{\lambda} \l \Re\rob{g} \r + \jmath \Mmod_{\lambda} \l \Im \rob{g} \r,\quad\lambda\in\R^+.
	\end{equation}
	We can now design a novel \ob \USF acquisition scheme for \LCT-bandlimited signals as illustrated in \figref{fig:lctarch}, where the folding (with $\lambda<1$) is implemented in continuous-time, prior to digitization, ensuring that the quantizer never overloads. However, this deviates from the signal model in \secref{sec:1BLCT} and necessitates new recovery approaches that can decode $g$ from modulo encoded $\qLLn$.
    	
    \begin{figure}[!t]
    	\centering
		\scalebox{0.5}{
			\tikzset{every picture/.style={line width=0.75pt}} 
			
			\begin{tikzpicture}[x=0.75pt,y=0.75pt,yscale=-1,xscale=1]
				
				\draw [line width=1.5]    (289.2,211.35) -- (396.8,211.35) ;
				\draw [shift={(400.8,211.35)}, rotate = 180] [fill={rgb, 255:red, 0; green, 0; blue, 0 }  ][line width=0.08]  [draw opacity=0] (13.4,-6.43) -- (0,0) -- (13.4,6.44) -- (8.9,0) -- cycle    ;
				\draw  [line width=1.5]  (110.6,211.35) .. controls (110.6,206.24) and (114.74,202.1) .. (119.85,202.1) .. controls (124.96,202.1) and (129.1,206.24) .. (129.1,211.35) .. controls (129.1,216.46) and (124.96,220.6) .. (119.85,220.6) .. controls (114.74,220.6) and (110.6,216.46) .. (110.6,211.35) -- cycle ; \draw  [line width=1.5]  (110.6,211.35) -- (129.1,211.35) ; \draw  [line width=1.5]  (119.85,202.1) -- (119.85,220.6) ;
				\draw  [fill={rgb, 255:red, 255; green, 255; blue, 255 }  ,fill opacity=1 ][line width=1.5]  (247.35,190.6) -- (292.35,190.6) -- (292.35,232.1) -- (247.35,232.1) -- cycle ;
				
				\draw [line width=1.5]    (171.4,211.35) -- (242.4,211.35) ;
				\draw [shift={(246.4,211.35)}, rotate = 180] [fill={rgb, 255:red, 0; green, 0; blue, 0 }  ][line width=0.08]  [draw opacity=0] (13.4,-6.43) -- (0,0) -- (13.4,6.44) -- (8.9,0) -- cycle    ;
				\draw [line width=1.5]    (420.1,211.35) -- (446.35,211.35) ;
				\draw [shift={(450.35,211.35)}, rotate = 180] [fill={rgb, 255:red, 0; green, 0; blue, 0 }  ][line width=0.08]  [draw opacity=0] (13.4,-6.43) -- (0,0) -- (13.4,6.44) -- (8.9,0) -- cycle    ;
				\draw [line width=1.5]    (410.7,220.75) -- (410.7,255.6) -- (162.2,255.35) -- (162.16,224.6) ;
				\draw [shift={(162.15,220.6)}, rotate = 89.92] [fill={rgb, 255:red, 0; green, 0; blue, 0 }  ][line width=0.08]  [draw opacity=0] (13.4,-6.43) -- (0,0) -- (13.4,6.44) -- (8.9,0) -- cycle    ;
				\draw [line width=1.5]    (410.57,172.1) -- (410.57,197.95) ;
				\draw [shift={(410.57,201.95)}, rotate = 270] [fill={rgb, 255:red, 0; green, 0; blue, 0 }  ][line width=0.08]  [draw opacity=0] (13.4,-6.43) -- (0,0) -- (13.4,6.44) -- (8.9,0) -- cycle    ;
				\draw [line width=1.5]    (130.5,211.35) -- (148.8,211.35) ;
				\draw [shift={(152.8,211.35)}, rotate = 180] [fill={rgb, 255:red, 0; green, 0; blue, 0 }  ][line width=0.08]  [draw opacity=0] (13.4,-6.43) -- (0,0) -- (13.4,6.44) -- (8.9,0) -- cycle    ;
				\draw  [fill={rgb, 255:red, 255; green, 255; blue, 255 }  ,fill opacity=1 ][line width=1.5]  (582.3,120.1) -- (662.8,120.1) -- (662.8,161.96) -- (582.3,161.96) -- cycle ;
				
				\draw [line width=1.5]    (680.33,141) -- (680.33,274.33) -- (119.5,274.33) -- (119.82,224.6) ;
				\draw [shift={(119.85,220.6)}, rotate = 90.37] [fill={rgb, 255:red, 0; green, 0; blue, 0 }  ][line width=0.08]  [draw opacity=0] (13.4,-6.43) -- (0,0) -- (13.4,6.44) -- (8.9,0) -- cycle    ;
				\draw [shift={(680.33,141)}, rotate = 90] [color={rgb, 255:red, 0; green, 0; blue, 0 }  ][fill={rgb, 255:red, 0; green, 0; blue, 0 }  ][line width=1.5]      (0, 0) circle [x radius= 4.36, y radius= 4.36]   ;
				\draw [line width=1.5]    (662.93,140.35) -- (708.2,140.35) ;
				\draw [shift={(712.2,140.35)}, rotate = 180] [fill={rgb, 255:red, 0; green, 0; blue, 0 }  ][line width=0.08]  [draw opacity=0] (13.4,-6.43) -- (0,0) -- (13.4,6.44) -- (8.9,0) -- cycle    ;
				\draw  [line width=1.5]  (152.9,211.35) .. controls (152.9,206.24) and (157.04,202.1) .. (162.15,202.1) .. controls (167.26,202.1) and (171.4,206.24) .. (171.4,211.35) .. controls (171.4,216.46) and (167.26,220.6) .. (162.15,220.6) .. controls (157.04,220.6) and (152.9,216.46) .. (152.9,211.35) -- cycle ; \draw  [line width=1.5]  (152.9,211.35) -- (171.4,211.35) ; \draw  [line width=1.5]  (162.15,202.1) -- (162.15,220.6) ;
				\draw  [fill={rgb, 255:red, 255; green, 255; blue, 255 }  ,fill opacity=1 ][line width=1.5]  (246.35,54.6) -- (291.35,54.6) -- (291.35,96.1) -- (246.35,96.1) -- cycle ;
				
				\draw [line width=1.5]    (185.6,211.85) -- (185.6,87.27) ;
				\draw [shift={(185.6,83.27)}, rotate = 90] [fill={rgb, 255:red, 0; green, 0; blue, 0 }  ][line width=0.08]  [draw opacity=0] (13.4,-6.43) -- (0,0) -- (13.4,6.44) -- (8.9,0) -- cycle    ;
				\draw [shift={(185.6,211.85)}, rotate = 270] [color={rgb, 255:red, 0; green, 0; blue, 0 }  ][fill={rgb, 255:red, 0; green, 0; blue, 0 }  ][line width=1.5]      (0, 0) circle [x radius= 4.36, y radius= 4.36]   ;
				\draw  [line width=1.5]  (176.27,74.02) .. controls (176.27,68.91) and (180.41,64.77) .. (185.52,64.77) .. controls (190.63,64.77) and (194.77,68.91) .. (194.77,74.02) .. controls (194.77,79.13) and (190.63,83.27) .. (185.52,83.27) .. controls (180.41,83.27) and (176.27,79.13) .. (176.27,74.02) -- cycle ; \draw  [line width=1.5]  (176.27,74.02) -- (194.77,74.02) ; \draw  [line width=1.5]  (185.52,64.77) -- (185.52,83.27) ;
				\draw  [line width=1.5]  (402.1,73.91) .. controls (402.1,69.24) and (405.89,65.44) .. (410.57,65.44) .. controls (415.25,65.44) and (419.04,69.24) .. (419.04,73.91) .. controls (419.04,78.59) and (415.25,82.38) .. (410.57,82.38) .. controls (405.89,82.38) and (402.1,78.59) .. (402.1,73.91) -- cycle ; \draw  [line width=1.5]  (404.58,67.92) -- (416.56,79.9) ; \draw  [line width=1.5]  (416.56,67.92) -- (404.58,79.9) ;
				\draw [line width=1.5]    (410.57,113.1) -- (410.57,86.38) ;
				\draw [shift={(410.57,82.38)}, rotate = 90] [fill={rgb, 255:red, 0; green, 0; blue, 0 }  ][line width=0.08]  [draw opacity=0] (13.4,-6.43) -- (0,0) -- (13.4,6.44) -- (8.9,0) -- cycle    ;
				\draw  [line width=1.5]  (401.3,211.35) .. controls (401.3,206.16) and (405.51,201.95) .. (410.7,201.95) .. controls (415.89,201.95) and (420.1,206.16) .. (420.1,211.35) .. controls (420.1,216.54) and (415.89,220.75) .. (410.7,220.75) .. controls (405.51,220.75) and (401.3,216.54) .. (401.3,211.35) -- cycle ; \draw  [line width=1.5]  (404.05,204.7) -- (417.35,218) ; \draw  [line width=1.5]  (417.35,204.7) -- (404.05,218) ;
				\draw  [line width=1.5]  (492.1,140.85) .. controls (492.1,135.74) and (496.24,131.6) .. (501.35,131.6) .. controls (506.46,131.6) and (510.6,135.74) .. (510.6,140.85) .. controls (510.6,145.96) and (506.46,150.1) .. (501.35,150.1) .. controls (496.24,150.1) and (492.1,145.96) .. (492.1,140.85) -- cycle ; \draw  [line width=1.5]  (492.1,140.85) -- (510.6,140.85) ; \draw  [line width=1.5]  (501.35,131.6) -- (501.35,150.1) ;
				\draw [line width=1.5]    (68,211.35) -- (107.6,211.35) ;
				\draw [shift={(111.6,211.35)}, rotate = 180] [fill={rgb, 255:red, 0; green, 0; blue, 0 }  ][line width=0.08]  [draw opacity=0] (13.4,-6.43) -- (0,0) -- (13.4,6.44) -- (8.9,0) -- cycle    ;
				\draw [line width=1.5]    (195.6,74.02) -- (243.27,74.02) ;
				\draw [shift={(247.27,74.02)}, rotate = 180] [fill={rgb, 255:red, 0; green, 0; blue, 0 }  ][line width=0.08]  [draw opacity=0] (13.4,-6.43) -- (0,0) -- (13.4,6.44) -- (8.9,0) -- cycle    ;
				\draw [line width=1.5]    (292.27,73.91) -- (397.2,73.91) ;
				\draw [shift={(401.2,73.91)}, rotate = 180] [fill={rgb, 255:red, 0; green, 0; blue, 0 }  ][line width=0.08]  [draw opacity=0] (13.4,-6.43) -- (0,0) -- (13.4,6.44) -- (8.9,0) -- cycle    ;
				\draw [line width=1.5]    (410.57,65.44) -- (410.57,23.6) -- (185.27,23.93) -- (185.49,60.77) ;
				\draw [shift={(185.52,64.77)}, rotate = 269.65] [fill={rgb, 255:red, 0; green, 0; blue, 0 }  ][line width=0.08]  [draw opacity=0] (13.4,-6.43) -- (0,0) -- (13.4,6.44) -- (8.9,0) -- cycle    ;
				\draw [line width=1.5]    (87,211.25) -- (87,10.75) -- (546,10.75) -- (546.1,127.6) ;
				\draw [shift={(546.1,131.6)}, rotate = 269.95] [fill={rgb, 255:red, 0; green, 0; blue, 0 }  ][line width=0.08]  [draw opacity=0] (13.4,-6.43) -- (0,0) -- (13.4,6.44) -- (8.9,0) -- cycle    ;
				\draw [shift={(87,211.25)}, rotate = 270] [color={rgb, 255:red, 0; green, 0; blue, 0 }  ][fill={rgb, 255:red, 0; green, 0; blue, 0 }  ][line width=1.5]      (0, 0) circle [x radius= 4.36, y radius= 4.36]   ;
				\draw  [line width=1.5]  (491.85,211.2) -- (450.1,238.81) -- (450.1,183.6) -- cycle ;
				
				\draw [line width=1.5]    (492.35,211.2) -- (501.2,211.2) -- (501.2,154.1) ;
				\draw [shift={(501.2,150.1)}, rotate = 90] [fill={rgb, 255:red, 0; green, 0; blue, 0 }  ][line width=0.08]  [draw opacity=0] (13.4,-6.43) -- (0,0) -- (13.4,6.44) -- (8.9,0) -- cycle    ;
				\draw [line width=1.5]    (420.04,73.91) -- (501.2,73.91) -- (501.34,127.6) ;
				\draw [shift={(501.35,131.6)}, rotate = 269.85] [fill={rgb, 255:red, 0; green, 0; blue, 0 }  ][line width=0.08]  [draw opacity=0] (13.4,-6.43) -- (0,0) -- (13.4,6.44) -- (8.9,0) -- cycle    ;
				\draw  [line width=1.5]  (536.85,140.85) .. controls (536.85,135.74) and (540.99,131.6) .. (546.1,131.6) .. controls (551.21,131.6) and (555.35,135.74) .. (555.35,140.85) .. controls (555.35,145.96) and (551.21,150.1) .. (546.1,150.1) .. controls (540.99,150.1) and (536.85,145.96) .. (536.85,140.85) -- cycle ; \draw  [line width=1.5]  (536.85,140.85) -- (555.35,140.85) ; \draw  [line width=1.5]  (546.1,131.6) -- (546.1,150.1) ;
				\draw [line width=1.5]    (509.6,140.85) -- (532.85,140.85) ;
				\draw [shift={(536.85,140.85)}, rotate = 180] [fill={rgb, 255:red, 0; green, 0; blue, 0 }  ][line width=0.08]  [draw opacity=0] (13.4,-6.43) -- (0,0) -- (13.4,6.44) -- (8.9,0) -- cycle    ;
				\draw [line width=1.5]    (553.35,141.25) -- (579.6,141.25) ;
				\draw [shift={(583.6,141.25)}, rotate = 180] [fill={rgb, 255:red, 0; green, 0; blue, 0 }  ][line width=0.08]  [draw opacity=0] (13.4,-6.43) -- (0,0) -- (13.4,6.44) -- (8.9,0) -- cycle    ;
				
				\draw [color={rgb, 255:red, 0; green, 0; blue, 0 }  ,draw opacity=1 ][line width=2.25]    (100,240.2) -- (111,240.2) ;
				
				\draw (347.18,48.16) node  [font=\LARGE]  {$u\sqb{n - 1}$};
				\draw (217.35,48.16) node  [font=\LARGE]  {$u\n$};
				\draw (742.72,140.66) node  [font=\LARGE]  {$\qLsecn$};
				\draw (347.18,190.5) node  [font=\LARGE]  {$x\sqb{n - 1}$};
				\draw (159.02,167.5) node  [font=\LARGE]  {$x\n$};
				\draw (389.52,140.5) node  [font=\Large]  {$\exp \l -\jmath \frac{a\rob{2n-1}\rob{\Tcur}^{2}}{2b} \r$};
				\draw (47.02,211.5) node  [font=\LARGE]  {$\gn$};
				\draw (465.29,208.29) node  [font=\LARGE]  {$c_{0}$};
				\draw (268.85,75.35) node  [font=\LARGE]  {$z^{-1}$};
				\draw (622.55,141.03) node  [font=\LARGE]  {$\csgn\rob{\cdot}$};
				\draw (269.6,211.35) node  [font=\LARGE]  {$z^{-1}$};

			\end{tikzpicture}
		}
		\caption{2nd-order \LSDQ described using state equations in \eqref{eq:sdq2-lct-lemma}.}
		\label{fig:lctsdq2}
	\end{figure}

	\begin{figure}[!t]
		\centering
		\scalebox{0.5}{
			
			\tikzset {_i39cmfvp6/.code = {\pgfsetadditionalshadetransform{ \pgftransformshift{\pgfpoint{0 bp } { 0 bp }  }  \pgftransformrotate{-270 }  \pgftransformscale{2 }  }}}
			\pgfdeclarehorizontalshading{_3ee2j9du4}{150bp}{rgb(0bp)=(1,1,1);
				rgb(37.5bp)=(1,1,1);
				rgb(62.5bp)=(0.98,0.94,0.93);
				rgb(100bp)=(0.98,0.94,0.93)}
			\tikzset{every picture/.style={line width=0.75pt}} 
			
			\begin{tikzpicture}[x=0.75pt,y=0.75pt,yscale=-1,xscale=1]
				
				\path  [shading=_3ee2j9du4,_i39cmfvp6] (106.2,29.98) .. controls (106.2,24.75) and (110.44,20.5) .. (115.68,20.5) -- (386.42,20.5) .. controls (391.66,20.5) and (395.9,24.75) .. (395.9,29.98) -- (395.9,147.58) .. controls (395.9,152.82) and (391.66,157.06) .. (386.42,157.06) -- (115.68,157.06) .. controls (110.44,157.06) and (106.2,152.82) .. (106.2,147.58) -- cycle ; 
				\draw  [color={rgb, 255:red, 235; green, 199; blue, 199 }  ,draw opacity=1 ][line width=3]  (106.2,29.98) .. controls (106.2,24.75) and (110.44,20.5) .. (115.68,20.5) -- (386.42,20.5) .. controls (391.66,20.5) and (395.9,24.75) .. (395.9,29.98) -- (395.9,147.58) .. controls (395.9,152.82) and (391.66,157.06) .. (386.42,157.06) -- (115.68,157.06) .. controls (110.44,157.06) and (106.2,152.82) .. (106.2,147.58) -- cycle ; 
				
				\draw  [color={rgb, 255:red, 235; green, 199; blue, 199 }  ,draw opacity=1 ][fill={rgb, 255:red, 249; green, 239; blue, 237 }  ,fill opacity=1 ][line width=3]  (212.78,6.14) .. controls (212.78,4.83) and (213.84,3.76) .. (215.15,3.76) -- (286.95,3.76) .. controls (288.26,3.76) and (289.32,4.83) .. (289.32,6.14) -- (289.32,35.59) .. controls (289.32,36.9) and (288.26,37.96) .. (286.95,37.96) -- (215.15,37.96) .. controls (213.84,37.96) and (212.78,36.9) .. (212.78,35.59) -- cycle ;
				
				\draw  [fill={rgb, 255:red, 255; green, 255; blue, 255 }  ,fill opacity=1 ] (225.1,59.31) -- (279.93,59.31) -- (279.93,99.31) -- (225.1,99.31) -- cycle ;
				
				\draw  [fill={rgb, 255:red, 255; green, 255; blue, 255 }  ,fill opacity=1 ] (423.7,59.31) -- (497.73,59.31) -- (497.73,99.31) -- (423.7,99.31) -- cycle ;
				
				\draw  [line width=1.5]  (155.17,79.56) .. controls (155.17,73.49) and (160.1,68.56) .. (166.17,68.56) .. controls (172.25,68.56) and (177.17,73.49) .. (177.17,79.56) .. controls (177.17,85.64) and (172.25,90.56) .. (166.17,90.56) .. controls (160.1,90.56) and (155.17,85.64) .. (155.17,79.56) -- cycle ; \draw  [line width=1.5]  (158.39,71.79) -- (173.95,87.34) ; \draw  [line width=1.5]  (173.95,71.79) -- (158.39,87.34) ;
				\draw [line width=1.5]    (166.17,118.56) -- (166.17,93.56) ;
				\draw [shift={(166.17,90.56)}, rotate = 90] [color={rgb, 255:red, 0; green, 0; blue, 0 }  ][line width=1.5]    (14.21,-4.28) .. controls (9.04,-1.82) and (4.3,-0.39) .. (0,0) .. controls (4.3,0.39) and (9.04,1.82) .. (14.21,4.28)   ;
				\draw [line width=1.5]    (177.9,79.56) -- (221.74,79.56) ;
				\draw [shift={(224.74,79.56)}, rotate = 180] [color={rgb, 255:red, 0; green, 0; blue, 0 }  ][line width=1.5]    (14.21,-4.28) .. controls (9.04,-1.82) and (4.3,-0.39) .. (0,0) .. controls (4.3,0.39) and (9.04,1.82) .. (14.21,4.28)   ;
				\draw  [line width=1.5]  (326.17,79.56) .. controls (326.17,73.49) and (331.1,68.56) .. (337.17,68.56) .. controls (343.25,68.56) and (348.17,73.49) .. (348.17,79.56) .. controls (348.17,85.64) and (343.25,90.56) .. (337.17,90.56) .. controls (331.1,90.56) and (326.17,85.64) .. (326.17,79.56) -- cycle ; \draw  [line width=1.5]  (329.39,71.79) -- (344.95,87.34) ; \draw  [line width=1.5]  (344.95,71.79) -- (329.39,87.34) ;
				\draw [line width=1.5]    (337.17,118.56) -- (337.17,93.56) ;
				\draw [shift={(337.17,90.56)}, rotate = 90] [color={rgb, 255:red, 0; green, 0; blue, 0 }  ][line width=1.5]    (14.21,-4.28) .. controls (9.04,-1.82) and (4.3,-0.39) .. (0,0) .. controls (4.3,0.39) and (9.04,1.82) .. (14.21,4.28)   ;
				\draw [line width=1.5]    (279.33,79.56) -- (323.17,79.56) ;
				\draw [shift={(326.17,79.56)}, rotate = 180] [color={rgb, 255:red, 0; green, 0; blue, 0 }  ][line width=1.5]    (14.21,-4.28) .. controls (9.04,-1.82) and (4.3,-0.39) .. (0,0) .. controls (4.3,0.39) and (9.04,1.82) .. (14.21,4.28)   ;
				\draw [line width=1.5]    (348.17,79.56) -- (419.9,79.56) ;
				\draw [shift={(422.9,79.56)}, rotate = 180] [color={rgb, 255:red, 0; green, 0; blue, 0 }  ][line width=1.5]    (14.21,-4.28) .. controls (9.04,-1.82) and (4.3,-0.39) .. (0,0) .. controls (4.3,0.39) and (9.04,1.82) .. (14.21,4.28)   ;
				\draw [line width=1.5]    (498.06,79.56) -- (541.9,79.56) ;
				\draw [shift={(544.9,79.56)}, rotate = 180] [color={rgb, 255:red, 0; green, 0; blue, 0 }  ][line width=1.5]    (14.21,-4.28) .. controls (9.04,-1.82) and (4.3,-0.39) .. (0,0) .. controls (4.3,0.39) and (9.04,1.82) .. (14.21,4.28)   ;
				\draw [line width=1.5]    (80.44,79.56) -- (152.17,79.56) ;
				\draw [shift={(155.17,79.56)}, rotate = 180] [color={rgb, 255:red, 0; green, 0; blue, 0 }  ][line width=1.5]    (14.21,-4.28) .. controls (9.04,-1.82) and (4.3,-0.39) .. (0,0) .. controls (4.3,0.39) and (9.04,1.82) .. (14.21,4.28)   ;

				\draw (373.17,52.25) node  [font=\Large]  {$\dnarr{y}\t$};
				\draw (302.17,56.25) node  [font=\Large]  {$y\t$};
				\draw (199.17,55.25) node  [font=\Large]  {$\uparr{g}\t$};
				\draw (337.57,137.25) node  [font=\Large]  {$\dnL\t$};
				\draw (166.17,135.25) node  [font=\Large]  {$\upL\t$};
				\draw (574.97,79.75) node  [font=\Large]  {$\qLLn$};
				\draw (35.61,80.06) node  [font=\Large]  {$g\t\in \BOmLCT$};
				\draw (460.72,79.31) node  [font=\Large]  {\LSDQ};
				\draw (252.52,79.31) node  [font=\Large]  {$\Mmod$};
				\draw (251.05,21) node  [font=\LARGE,color={rgb, 255:red, 215; green, 35; blue, 38 }  ,opacity=1 ]  {\LM};

			\end{tikzpicture}
		}
		\caption{Proposed \ob \USF architecture in \LCT (\MLSDQ).}
		\label{fig:lctarch}
	\end{figure}

	\begin{SCfigure}
		\centering
		\includegraphics[width=0.6\textwidth]{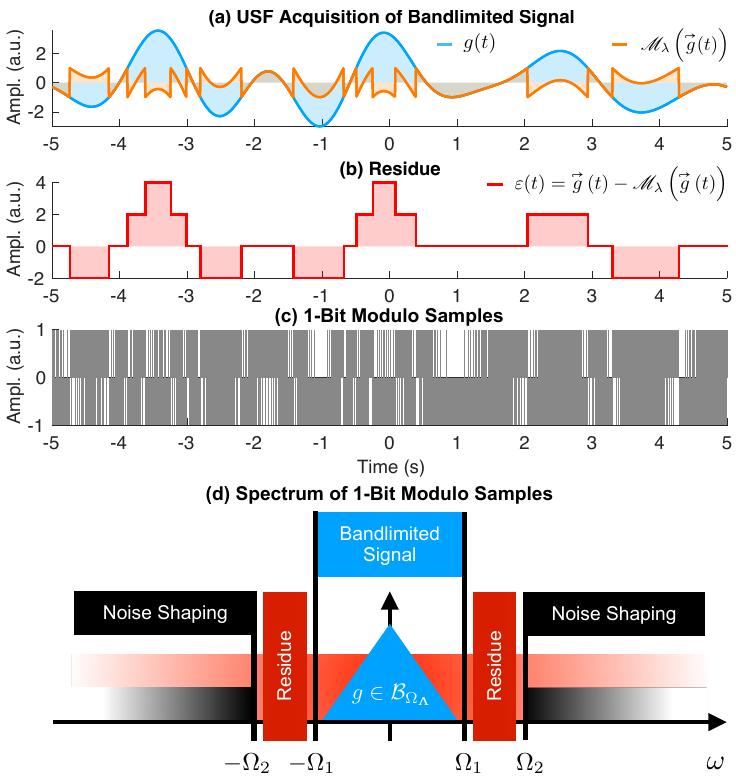}
		\caption{Modular decomposition property for $\Lmtx = \LmtxFT$ (Fourier transform). (a) USF acquisition of real-valued bandlimited signal $g\t$. (b) Residue obtained by expressing $\res\t$ from \eqref{eq:lct-mod-dec}. (c) \ob modulo samples $\qLLn$ acquired by \MLSDQ. (d) Deconstructing the spectrum of $\qLLn$ in \eqref{eq:qLn-dny}.}
		\label{fig:noiseshaping}
	\end{SCfigure}

	\bpara{Problem Statement.} Our goal is to recover  $g\in \BOmLCT$ from \ob modulo samples $\qLLn$ acquired via \MLSDQ.
	
	\section{Recovery Algorithm}
	In this section, we develop a new recovery algorithm for \MLSDQ, which jointly capitalizes on the transform domain separation of the signal and folds together with noise-shaping of quantization noise. The proposed algorithm is agnostic to the folding threshold ($\lambda$), unlike the time-domain methods proposed in \cite{Graf:2019:C, Florescu:2022:Ja, Eamaz:2024:J}. Although we use $\qLLn$ in our derivations, the same approach works with $\qLL^{[l]}\n$ captured by $l$th\nobreakdash-order schemes. The starting point for the recovery of $g\t \in \BOmLCT$ is the modular decomposition property \cite{Bhandari:2017:C}, which applies to arbitrary signals and enables us to write,
	\begin{align}
		g\t &= \dnarr{y}\t + \dnarr{\res}\t\label{eq:lct-mod-dec}\\
		&= \dnL\t \l \Mmod_{\lambda}\rob{\uparr{g}\t} + \sum\limits_{k=0}^{K-1}\ck \ind\rob{t - \nk \Tcur} \r.\notag
	\end{align}
	We explain the modular decomposition property in \figref{fig:noiseshaping} on the Fourier transform ($\Lmtx = \LmtxFT$) example. The input signal $g$ to \MLSDQ and the modulo signal $y$ are shown in \figref{fig:noiseshaping}(a). Residue $\res = \uparr{g} - y$ to be estimated is depicted in \figref{fig:noiseshaping}(b). One-bit modulo samples $\qLLn$ captured by \MLSDQ are shown in \figref{fig:noiseshaping}(c). Due to the superposition of $\uparr{g}$ and $\res$ in the \LCT domain, we expect to observe a separation as shown in \figref{fig:noiseshaping}(d), without the effect of quantization noise. This separation property has been previously leveraged in the case of Fourier  \cite{Bhandari:2021:J} and \LCT domain \cite{Zhang:2023:C}, respectively. In our case, this separation does not apply directly because of the contamination of quantization noise resulting from 1-Bit samples. However, by controlled oversampling, the noise-shaping property of \LSDQ can be utilized to our advantage to estimate the residue in the interval between $\Omega_1 = \sfrac{2 \pi M_1 b}{\tau}$ and $\Omega_2= \sfrac{2 \pi M_2 b}{\tau}$, shown in \figref{fig:noiseshaping}(d). Hence, despite the presence of heavy distortion arising from 1-Bit samples, we can still repurpose the transform domain separation. 

	Rewriting \eqref{eq:un-lct-proof} with $\dnarr{y}\t$ as an input produces,
	\begin{equation}
		\qLLn = \dnarr{y}\rob{n\Tcur} + e^{-\jmath\frac{a \rob{2n - 1} \rob{\Tcur}^2}{2b}} u \sqb{n - 1} - \un.
		\label{eq:qLn-dny}
	\end{equation}
	Next, we substitute $\dnarr{y}\t$ from  \eqref{eq:lct-mod-dec} in \eqref{eq:qLn-dny} and obtain
	\begin{align}
		\qLLn &= g\n - \dnarr{\res}\n + e^{-\jmath\frac{a \rob{2n - 1} \rob{\Tcur}^2}{2b}} u \sqb{n - 1} - \un \notag\\
		&= g\n - \dnarr{\res}\n + r_\Lmtx\n.
		\label{eq:qln-mdp-expanded}
	\end{align}
	By definition, $\res\n$ is a simple function and hence, its first-order difference $\diff{\res}\n$ results in a sparse representation with support $\{\nk\}_{k=0}^{K-1}$ and weights $\{\ck\}_{k=0}^{K-1}$. In what follows, we will elucidate that $\diff{\res}\n$ maps to a sum of complex exponentials in the \LCT domain, enabling their parameter estimation using known spectral estimation methods. To capitalize on this insight, we proceed by applying $\vLn$ to \eqref{eq:qln-mdp-expanded}, resulting in
	\begin{align}
		z\n &= \l \qLL \LCTconv \vL \r \n = \dnL \n \l \uparr{q}_{\Lmtx,\lambda} \conv v \r \n \notag\\
		&= \dnL\n \l \diam{g}_\Lmtx\n - \diff{\res}_\Lmtx\n + \diam{r}_\Lmtx \n \r
		\label{eq:wn-def}
	\end{align}
	where $\diam{o}_\Lmtx \DE \diff{(\uparr{o})}$. We decompose the sum in \eqref{eq:wn-def} as follows,
	\begin{equation}
		\begin{cases}
			& \dnL\n \diam{g}_\Lmtx\n\\
			-&\dnL\n\diff{\res}_\Lmtx\n = -\dnL\n\sum\nolimits_{k=0}^{K-1} \ck \delta\sqb{n - \nk}\\
			& \dnL\n\diam{r}_\Lmtx \n
		\end{cases}.
	\end{equation}
	Note that $\diam{g}_\Lmtx \in \BOmLCT \Rightarrow \dnL \diam{g}_\Lmtx \in \BOmLCT$. The second term $\dnL\diff{\res}_\Lmtx$ is completely parametrized by $\{\ck, \nk \}_{k=0}^{K-1}$. The last term $\dnL\diam{r}_\Lmtx$ is the effect of noise shaping. We can now analyze the constituent components in the transform domain.
	
	\bpara{Spectral Separation.} Let $\LCTop{z}\m$ denote the \DLCT of $z\n$ as defined in \eqref{eq:dlct-def}. We can partition $\LCTop{z}\m$ as follows
	\begin{equation}
		\LCTop{z}\m = 
		\begin{cases}
			\diam{\LCTop{g}}\m - \LCTop{\diff{\res}}\m + \diam{\LCTop{r}}\m &\abs{m} < M_1\\
			-\LCTop{\diff{\res}}\m + \diam{\LCTop{r}}\m &M_1 \leqslant \abs{m} < M_2\\
			- \LCTop{\diff{\res}}\m + \diam{\LCTop{r}}\m &M_2 \leqslant \abs{m}
		\end{cases}.
		\label{eq:wm-dlct}
	\end{equation}
	Spectral separation of $\LCTop{z}\m$ is shown in \figref{fig:noiseshaping}(d).

	\begin{algorithm}[!t]
		\caption{\ob USF beyond Fourier Domain.}
		\label{alg:algo-sum}
		\begin{algorithmic}[1]
			\renewcommand{\algorithmicrequire}{\textbf{Input:}}
			\renewcommand{\algorithmicensure}{\textbf{Output:}}
			\REQUIRE \ob modulo samples $\qLLn$ or $\qLLsecn$, number of folds $K$, sampling period $\Tcur$, bounds $M_1$, $M_2$ and $M_3$.
			\STATE Apply difference filter $\vLn$ to $\qLn$ as shown in \eqref{eq:wn-def}.
			\STATE Calculate \DLCT of $z\n$ using \eqref{eq:dlct-def}.
			\STATE Calculate $\LCTop{f}\m$ as shown in \eqref{eq:fm-def}.
			\STATE Solve \eqref{eq:ann-eq} to estimate $\left \{ \nk \right \}_{k=0}^{K-1}$.
			\STATE Obtain $\left \{ \ck \right \}_{k = 0}^{K - 1}$ from \eqref{eq:fm-def} using least-squares with $m \in [M_1, M_3]$.
			\STATE Reconstruct residue using $\tilde{\res}\n = \sum\nolimits_{k=0}^{n - 1} \diff{\res}\sqb{k}$.
			\STATE Recover $\qMBn = \qLn + \dnL\n  \tilde{\res}\n$.
			\STATE Low-pass filter $\qMBn$ to get $\tilde{g}\t$ as shown in \eqref{eq:mod-1b-low-pass}.
			\ENSURE  Signal reconstruction $\tilde{g}\t$.
		\end{algorithmic}
	\end{algorithm}

	\bpara{Residue Parameter Estimation.} The effect of noise shaping in \eqref{eq:wm-dlct} is that the contribution of $\diam{r}_\Lmtx$ in the interval between $0$ and $M_2$ is negligible. This enables the isolation of $\LCTop{\diff{\res}}\m$ from \eqref{eq:wm-dlct} in the interval between $M_1$ and $M_2$. We use this insight to estimate the residue. Note that by definition,
	\begin{equation}
		\LCTop{\diff{\res}}\m = \Tcur \sqrt{\frac{-\jmath}{\tau}} \sum\limits_{n=0}^{N-1} \dnarr{\diff{\res}}\n \kLinv \rob{m b \omega_0, n\Tcur}.
		\label{eq:z-lct-m}
	\end{equation}
	Then, by modulating \eqref{eq:z-lct-m} with $\frac{\sqrt{\tau}}{\Tcur\sqrt{-\jmath}} e^{-\jmath \frac{d\rob{m b \omega_0}^2}{2b}}$, we get,
	\begin{align}
		\LCTop{f}\m &= \frac{\LCTop{\diff{\res}}\m \sqrt{\tau}}{\Tcur \sqrt{-\jmath}} e^{-\jmath \frac{d\rob{m b \omega_0}^2}{2b}} = \sum\limits_{k=0}^{K-1} \ck e^{-\jmath \nk T m \omega_0}
		\label{eq:fm-def}
	\end{align}
	which is a \emph{sum of exponentials}. Hence, the frequencies corresponding to folding instants $\{ \nk \}_{k=0}^{K-1}$ can be estimated using Prony's method. To this end, let us denote by $\nvec{h}$ a $(K + 1)$-tap filter with z\nobreakdash-transform $H\rob{z} = \sum\nolimits_{n=0}^{K}h\n z^{-n} = \prod\nolimits_{k = 0}^{K - 1}\l 1 - r_k z^{-1}\r$ where $r_k = e^{-\jmath\nk T \omega_0}$. It is well known that $\mat{h}$ annihilates $\LCTop{f}\sqb{l},l \in [M_1 + K, M_2 - 1]$ \cite{Stoica:1997:B}, since
	\begin{equation}
		\l h \conv \LCTop{f} \r \sqb{l} = \sum\nolimits_{k=0}^{K - 1} c\sqb{k} r_{k}^{l - M_1} \sum\nolimits_{p=0}^{K} h\sqb{p}r_{k}^{-p} = 0.
		\label{eq:ann-eq}
	\end{equation}
	This can be algebraically rewritten as $\mat{T}(\nvec{\LCTop{f})}\nvec{h} = \nvec{0}$, where $\mat{T}({\nvec{\LCTop{f}}})$ is a $\rob{M_2 - M_1 - K} \times \rob{K + 1}$ Toeplitz matrix constructed from $\left \{ \LCTop{f}\m \right \}_{m=M_1}^{M_2 - 1}$ with length $|\Mn| = M_2 - M_1$ where $\Mn = [M_1,M_2].$ This system of equations can be solved when $|\Mn| \geqslant 2K$, which leads to the estimation of the filter $\mat{h}$. Folding locations $\left \{ \nk \right \}_{k=0}^{K-1}$ are obtained from filter roots $\{r_k\}_{k=0}^{K-1}$. Prony's method is known to be sensitive to perturbations, and robust solutions to \eqref{eq:ann-eq} can be achieved using \emph{high-resolution spectral estimation} methods, such as the Matrix Pencil Method (MPM) \cite{Hua:1990:J}. With $\{n_k\}_{k=0}^{K-1}$ known, the amplitudes $\left \{ \ck \right \}_{k}^{K-1}$ are estimated using least-squares (LS) inversion of the system of equations in \eqref{eq:fm-def}.

	\begin{SCfigure}
		\centering
		\includegraphics[width=0.6\textwidth]{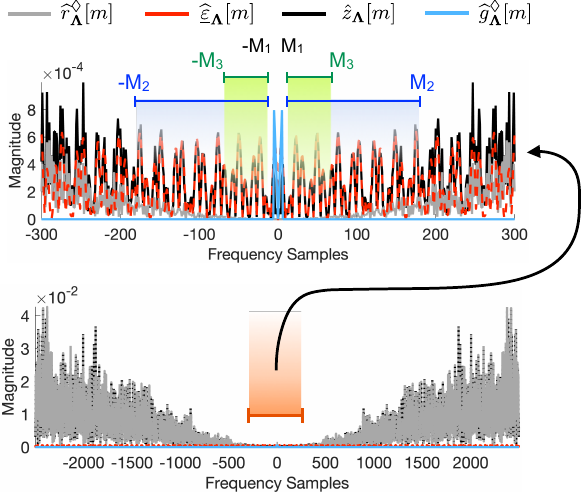}
		\caption{Spectral separation of $\LCTop{z}\m$ in \eqref{eq:wm-dlct} with bandwidth $M = 5$ and $\Lmtx = \LFT$ (Fourier transform). The input is captured by 1st-order \MLSDQ with $\lambda = 0.9$ and downchirped modulo signal $\dnarr{y}\t$ has $K = 4$ folds. $M_1$ and $M_2$ bound the interval for the estimation of $\{\nk\}_{k=0}^{K-1}$, and the interval between $M_1$ and $M_3$ is used to estimate $\{\ck\}_{k=0}^{K-1}$.}
		\label{fig:tworegions}
	\end{SCfigure}

	\bpara{Enhanced Estimation of Residue.} The estimates of $\{ \ck  \}_{k=0}^{K-1}$ rely on the estimates $\{ \nk  \}_{k=0}^{K-1}$, but require only $K$ equations. Let $\Mc = \left[ M_1,M_2 \right]$, then, as shown in \figref{fig:tworegions}, we have observed the estimation using a smaller interval $ M_3\in \mathsf{I} \subseteq \Mc, |\mathsf{I}|\geq K$ is an empirically superior approach.

	\bpara{Bandlimited Signal Reconstruction.} To recover $g\t$, we first obtain an estimate $\tilde{\res}\n$ via the anti-difference operator such that, $\tilde{\res}\n = \sum\nolimits_{k=0}^{n - 1} \diff{\res}\sqb{k}, n\geq 1$ where $\diff{\res}\n = \sum\nolimits_{k = 0}^{K-1}\ck \delta\sqb{n - \nk}$ with $\left \{ \ck, \nk \right \}_{k=0}^{K-1}$ obtained in the previous step. Adding demodulated residue estimate $\tilde{\res}\n$ to $\qLLn$ gives multi-bit samples $\qMBn = \qLLn + \dnL\n  \tilde{\res}\n$. To obtain the full reconstruction of $\tilde{g}\t$, we filter $\qMBn$ with the interpolation kernel $\idker$ with bandwidth $\Omega$
	\begin{equation}
		\tilde{g}\t = \osrat e^{-\jmath\frac{a t ^ 2}{2b}} \sum\limits_{n \in \Z} \uparr{q}_{\sf{MB}} \n \idker_{\OmL} \rob{\frac{t}{T} - \frac{n}{\osfact}}.
		\label{eq:mod-1b-low-pass}
	\end{equation}
	The designed algorithm is summarized in Algorithm \ref{alg:algo-sum}.

	\section{Numerical Experiments}

	\begin{SCfigure}
		\centering
		\includegraphics[width=0.65\textwidth]{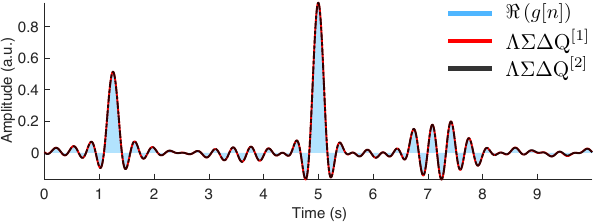}
		\caption{Reconstruction of a signal bandlimited in Fresnel Transform domain ($\Lambda_{\textsf{FrT}}, b = 2$) with $M = 30$ (Experiment 6 in \tabref{tab:lsdq-exps}). Reconstruction MSE achieved by 1st-order \LSDQ\;(\LSDQone) is $1.56 \times 10^{-6}$ and MSE achieved by 2nd-order \LSDQ\;(\LSDQtwo) is $4.33 \times 10^{-7}$.}
		\label{fig:lsdqexp6}
	\end{SCfigure}

	\bgroup
	\def\arraystretch{1.1}%
	\begin{table}[!t]
		\caption{\LSDQ Numerical Experiments.}
		\centering
		\resizebox{0.5\textwidth}{!}{%
			\begin{tabular}{|c|ccccc|}
				\hline
				Exp. & $\Lmtx$ & $M$ & $\osrat$ & $\MSEOp_1\rob{\nvec{g}, \nvec{\tilde{g}}}$ & $\MSEOp_2\rob{\nvec{g}, \nvec{\tilde{g}}}$ \\ \hline
				& & & $\times 10^{-3}$ & $\times 10^{-6}$ & $\times 10^{-7}$ \\ \hline
				$1$ & $\Lambda_{\mathsf{FT}}$ & $10$ & $6.67$ & $0.82$ & $6.79$\\ \hline
				$2$ & $\Lambda_{\theta}, \theta=\sfrac{\pi}{3}$ & $10$ & $6.67$ & $2.49$ & $6.95$\\ \hline
				$3$ & $\Lambda_{\mathsf{FrT}}, b=1$ & $10$ & $6.67$ & $2.14$ & $4.87$\\ \hline
				$4$ & $\Lambda_{\mathsf{FT}}$ & $30$ & $5.00$ & $0.71$ & $4.23$\\ \hline
                $5$ & $\Lambda_{\theta}, \theta=\sfrac{\pi}{16}$ & $30$ & $5.00$ & $1.28$ & $5.65$\\ \hline
                $6$ & $\Lambda_{\mathsf{FrT}}, b=2$ & $30$ & $5.00$ & $1.56$ & $4.33$\\ \hline
		\end{tabular}}%
		\label{tab:lsdq-exps}
	\end{table}
	\egroup		
	
\bgroup
	\def\arraystretch{1.2}%
	\begin{table*}[!t]
		\caption{\MLSDQ Numerical Experiments -- $M = 10$, $\MSEOp_{1}\rob{\nk, \tilde{n}_k} = 0$, $\MSEOp_{2}\rob{\nk, \tilde{n}_k} = 0$ (due to rounding) for all experiments.}
		\centering
		\resizebox{\textwidth}{!}{%
			\begin{tabular}{|c|ccccccccccccc|}
				\hline
				Exp. & $\Lambda_{LS}$ & $M_2 - M_1$ & $M_3 - M_1$ & $\osfact$ & $\infnorm{g}$ & $\lambda$ & $K$ & $\MSEOp_{1}\rob{\ck, \tilde{c}_k}$ & $\MSEOp_{1}\rob{\res, \tilde{\res}}$ & $\MSEOp_{1}\rob{g, \tilde{g}}$ & $\MSEOp_{2}\rob{\ck, \tilde{c}_{k}}$ & $\MSEOp_{2}\rob{\res, \tilde{\res}}$ & $\MSEOp_{2}\rob{g, \tilde{g}}$ \\ \hline
				& & & & $\times 10^{-4}$ & & & & $\times 10^{-7}$ & $\times 10^{-4}$ & $\times 10^{-6}$ & $\times 10^{-9}$ & $\times 10^{-4}$ & $\times 10^{-7}$\\ \hline
                $1$ & $\Lambda_{\mathsf{FT}}$ & $400$ & $80$ & $2.86$ & $1.90$ & $0.75$ & $10$ & $0.54$ & $3.06$ & $0.43$ & $8.79$ & $3.06$ & $1.19$\\ \hline
                $2$ & $\Lambda_{\theta}, \theta = \sfrac{\pi}{3}$ & $500$ & $80$ & $2.50$ & $2.73$ & $0.85$ & $18$ & $2.30$ & $6.21$ & $2.39$ & $5.12$ & $6.19$ & $2.33$\\ \hline
                $3$ & $\Lambda_{\theta}, \theta = \sfrac{\pi}{4}$ & $500$ & $120$ & $2.00$ & $3.74$ & $0.85$ & $22$ & $1.61$ & $6.12$ & $6.35$ & $8.46$ & $6.06$ & $2.26$\\ \hline
                $4$ & $\Lambda_{\theta}, \theta = \sfrac{\pi}{16}$ & $800$ & $200$ & $1.25$ & $6.69$ & $0.80$ & $42$ & $1.77$ & $6.44$ & $4.04$ & $6.43$ & $6.40$ & $2.99$\\ \hline
                $5$ & $\Lambda_{\mathsf{FrT}}, b = 1$ & $600$ & $150$ & $1.72$ & $2.83$ & $0.80$ & $16$ & $2.25$ & $3.38$ & $2.03$  & $4.55$ & $3.36$ & $1.03$\\ \hline
                $6$ & $\Lambda_{\mathsf{FrT}}, b = 2$ & $750$ & $180$ & $1.25$ & $3.83$ & $0.80$ & $38$ & $2.12$ & $5.80$ & $1.32$ & $6.83$ & $5.79$ & $2.61$\\ \hline
                $7$ & $\Lambda_{\mathsf{FrT}}, b = 3$ & $800$ & $200$ & $1.00$ & $5.49$ & $0.80$ & $42$ & $1.11$ & $5.13$ & $1.10$ & $1.94$ & $5.12$ & $2.08$\\ \hline
			\end{tabular}%
		}
		\label{tab:lct-sdq-num}
	\end{table*}
	\egroup

	\begin{figure*}[!t]
		\centering
		\includegraphics[width=0.98\linewidth]{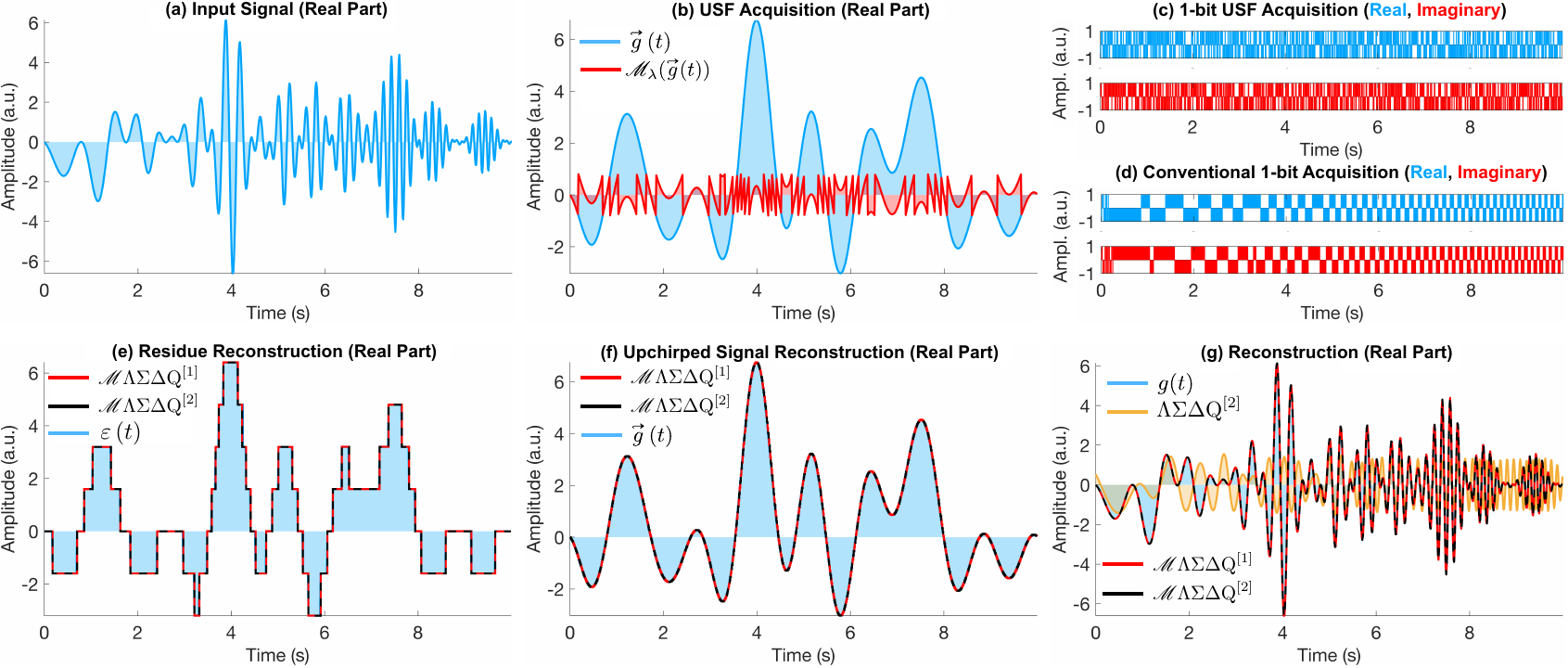}
		\caption{Visualization of acquisition (a)-(d) and reconstruction (e)-(g) for \MLSDQ for Experiment 4 in \tabref{tab:lct-sdq-num}. The 1st-order scheme is denoted by \MLSDQone, 2nd-order scheme by \MLSDQtwo, and \LSDQtwo refers to the second-order conventional scheme. (a) Input \FrFT-bandlimited signal ($\Lambda_\theta, \theta = \sfrac{\pi}{16}$). (b) \USF acquisition of the upchirped signal. (c) \ob modulo samples acquired using \MLSDQtwo. (d) Conventional \ob samples acquired using \LSDQtwo. (e) Residue reconstruction from \ob modulo samples. (f) Reconstruction of upchirped signal. (g) Recovery of the \HDR signal.}
		\label{fig:frftvisualexp}
	\end{figure*}

		\begin{SCfigure}
		\centering
		\includegraphics[width=0.6\textwidth]{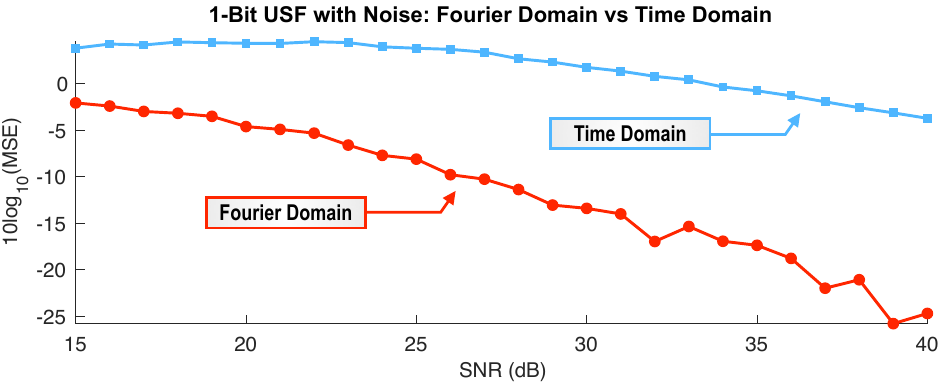}   
		\caption{Proposed Fourier-domain algorithm outperforms time-domain approach from \cite{Graf:2019:C} when the input signal is corrupted by noise. Setup parameters -- trials per each point $P = 1000$, folding threshold $\lambda = 1$, $\infnorm{g} = 4.5$ with $K = 8$ folds. Oversampling ratio $\osrat = 8\times 10^{-4}$ was determined using the criterion in \cite{Graf:2019:C}. Parameters for Fourier-domain method: $M_1 = 11$, $M_2 = 251$ and $M_3 = 50$.}
		\label{fig:noisyreccomparison}
	\end{SCfigure}

	We validate the proposed architectures and algorithm through numerical experiments using LCT\nobreakdash-bandlimited signals. We begin by applying 1st and 2nd-order \LSDQ and compare the reconstruction performance for signals with bounded inputs. Next, we conduct experiments with the \MLSDQ architecture using signals that significantly exceed the quantizer's threshold. We also examine how outband intervals $\Mn$ and $\Mc$, respectively, and the oversampling ratio $\osrat$ affect reconstruction performance. For the Fourier Transform case, we compare our recovery algorithm with the time-domain algorithm proposed in \cite{Graf:2019:C} and demonstrate the empirical robustness of our algorithm.

	\bpara{A. LCT Domain \ob Sampling.} We generate $g \in \BOmLCT$ with $\infnorm{g} \leqslant 1$ and pass it into the 1st and 2nd-order \LSDQ schemes. We perform numerical experiments using three different transforms and present the results in \tabref{tab:lsdq-exps}. The metric $\MSEOp_1(\nvec{g}, \nvec{\tilde{g}})$ represents the MSE for reconstruction of the 1st\nobreakdash-order scheme, while $\MSEOp_2(\nvec{g}, \nvec{\tilde{g}})$ denotes the MSE achieved by the 2nd\nobreakdash-order scheme. As expected, the 2nd\nobreakdash-order scheme reconstructs the signal with lower MSE due to its superior noise-shaping capability. The results for the real part from Experiment 6 are illustrated in \figref{fig:lsdqexp6}.

	\bpara{B. LCT Domain \ob Unlimited Sampling.} With $\infnorm{g} > 1$, non-USF \LSDQ (in \secref{sec:1BLCT}) would saturate, and hence conventional recovery fails. In such scenarios, we utilize our novel \MLSDQ architecture and demonstrate its capability through simulations for three transforms:  Fourier Transform (\FT), Fractional Fourier Transform (\FrFT), and Fresnel Transform (\Fr). We generate an input signal $g[n]$ such that $g \in \BOmLCT$ and $\infnorm{g} > 1$. The results are presented in \tabref{tab:lct-sdq-num}. Since the 2nd\nobreakdash-order \LSDQ offers better noise rejection, it outperforms the setup using the 1st\nobreakdash-order scheme. Experiment 4 is illustrated in \figref{fig:frftvisualexp}—the acquisition is shown in \figref{fig:frftvisualexp}(a)-(d), and the reconstruction is depicted in \figref{fig:frftvisualexp}(e)-(g).
	
	\begin{SCfigure}
		\centering
		\includegraphics[width=0.6\textwidth]{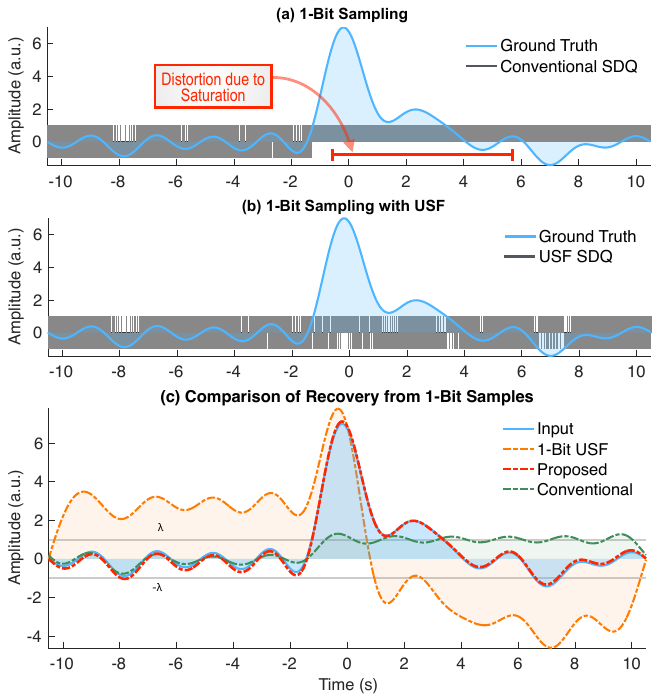}
		\caption{Recovery from undersampled 1-Bit modulo samples captured by \MLSDQone  with $\lambda = 1$ -- proposed Fourier-domain approach recovers input signal with $19\times$ smaller oversampling factor than the guarantee in \cite{Graf:2019:C}. The conventional \SDQ is overloaded, and the signal reconstruction is not possible.}
		\label{fig:tdvsfd}
	\end{SCfigure}
	
	\begin{figure*}[!t]
		\centering
		\includegraphics[width=\linewidth]{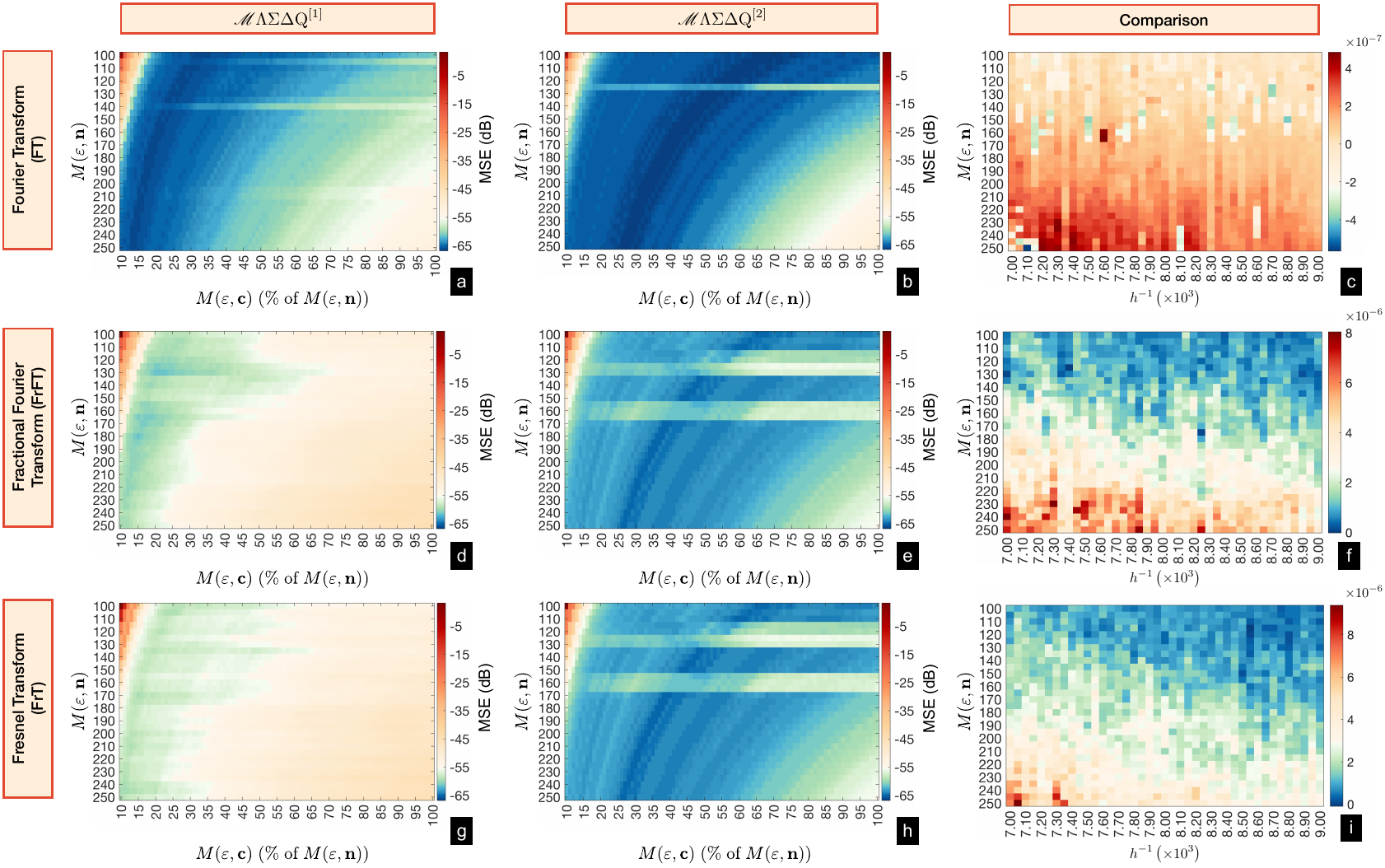}
		\caption{Dependence of reconstruction MSE of a signal with $M = 4$ on outband intervals $\Mn$, $\Mc$ and oversampling ratio $\osrat$ for \MLSDQ\;(1st-order scheme is denoted by \MLSDQone and 2nd-order scheme by \MLSDQtwo). Common parameters: $\lambda = 0.75$, $K = 10$ folds, $P = 100$ trials, $\Lambda_{\theta}, \theta=\sfrac{\pi}{3}$ for \FrFT and $b = 2$ for \Fr. The first two columns show average MSE for different $\Mn$ and $\Mc$ with $\osrat = 1.25 \times 10^{-4}$. For each $\Mn$, we sweep $\abs{\Mc}$ from $10\%$ to $100\%$ of $\abs{\Mn}$. Results for \MLSDQone (a), (d) and (g) suggest that the optimal $0.1 \abs{\Mn} \leq \abs{\Mc} \leq 0.4 \abs{\Mn}$ for most $\Mn$. For \MLSDQtwo (b), (e) and (h), the optimal choice of $0.15 \abs{\Mn} \leq \abs{\Mc} \leq 0.6 \abs{\Mn}$. The last column (c), (f) and (i) shows the dependence of the difference of average MSE for \MLSDQone and \MLSDQtwo on $\osrat$ and $\Mn$ with $\abs{\Mc} = 0.4 \abs{\Mn}$. As plots show, \MLSDQtwo is better, especially with lower oversampling.
		}
		\label{fig:foldcountmap}
	\end{figure*}

	\bpara{C. Comparison with \ob \USF in Time-Domain.} We compare our method with the time\nobreakdash-domain approach developed in \cite{Graf:2019:C}. Let us denote the noisy input as
	$$\bar{g}\n = g\n + w\n, \quad w \sim \mathcal{N}\rob{0,\sigma^2}$$
	where $w$ are i.i.d. samples drawn from a  Gaussian distribution with zero mean and variance $\sigma^2$. We omit chirp modulation since $\Lmtx=\LmtxFT$. We choose $\osrat$ to satisfy the recovery criterion given in \cite{Graf:2019:C} and set $\sigma^2 = P_s 10^{-\SNR / 10}$ where $P_s = {N^{-1}}\sum\nolimits_{n=0}^{N-1}\abs{g\n}^2$. We compare both algorithms based on average MSE, or
	$\MSEOp_{\textsf{avg}} = \frac{1}{PN}\sum\nolimits_{p = 0}^{P - 1}\sum\nolimits_{n=0}^{N-1} \abs{g_p\n - \tilde{g}_p\n}^2$ %
	where $P$ is the number of trials, $g_p\n$ denotes a randomly generated ground truth signal,
	\[
	g_p\n = \frac{\infnorm{g} \Re \l \Linv \sqb{\LCTop{g}\m} \r}{\infnorm{\Re \l \Linv \sqb{\LCTop{g}\m} \r}},\quad\LCTop{g}\m \in \Unif\l 0, 1\r
	\]
	where $\abs{m} \leq M$, $\Lmtx = \LmtxFT$ and $\infnorm{g}$ is the desired amplitude. Furthermore, $\tilde{g}_p\n$ is the corresponding signal reconstruction. \figref{fig:noisyreccomparison} shows the comparison of the two methods, where we used a second-order B-spline as a filtering kernel for the time\nobreakdash-domain method \cite{Graf:2019:C}. 
	We implement our method in Alg.~\ref{alg:algo-sum} with cyclic differences, and for spectral estimation (Step 4), we use the MPM technique \cite{Hua:1990:J}. As seen in \figref{fig:noisyreccomparison}, our transform-domain method achieves better reconstruction.  In contrast, the time\nobreakdash-domain approach relies on thresholding based on local information, which proves less stable in noisy scenarios and struggles with dense signal folds. We demonstrate that the proposed method reduces the required oversampling factor ($\sfrac{\osrat}{\osrat_{\textsf{TD}}}$). Results from three different experiments are presented in \tabref{tab:und-td-vs-fd}. The final experiment, illustrated in \figref{fig:tdvsfd}, shows that the proposed approach successfully reconstructs the input signal with $19\times$ lower oversampling compared to the time-domain method.
    
	\bgroup
	\def\arraystretch{1.2}%
	\begin{table}[!]
		\caption{Proposed Fourier-domain method outperforms time-domain approach of \cite{Graf:2019:C} when reconstructing from undersampled \ob modulo samples.}
		\centering
		\resizebox{0.5\textwidth}{!}{%
			\begin{tabular}{|c|cccccc|}
				\hline
				Exp. & $K$ & $\Mn$ & $\Mc$ & $\sfrac{\osrat}{\osrat_\textsf{TD}}$ & $\MSEOp_{\textsf{TD}}\rob{\nvec{g}, \nvec{\tilde{g}}}$ & $\MSEOp_{\textsf{FD}}\rob{\nvec{g}, \nvec{\tilde{g}}}$ \\ \hline
				$1$ & $2$ & $35$ & $7$ & $26.78$ & $0.53$ & $2.27\times 10^{-3}$ \\\hline
				$2$ & $4$ & $110$ & $38$ & $20.89$ & $0.96$ & $7.57\times 10^{-3}$ \\\hline
				$3$ & $8$ & $200$ & $32$ & $19.67$ & $8.23$ & $1.11\times 10^{-2}$ \\\hline
		\end{tabular}}%
		\label{tab:und-td-vs-fd}
	\end{table}
	\egroup
	
	\bpara{D. Effect of Outband Intervals $\Mn$ and $\Mc$.} To investigate the impact of $\{\Mn,\Mc\}$ on the reconstruction, we vary the length of $\{\Mn,\Mc\}$ for $P$ randomly generated inputs. For an $l$th-order quantizer, $\{\Mn,\Mc\}$, we calculate the average MSE
	\begin{equation*}
		\MSElnc{l} = \frac{1}{P}\sum\nolimits_{p=0}^{P-1}\MSEOp_l\rob{\nvec{g}_{p}, \tilde{\nvec{g}}_{p}\rnc}
	\end{equation*}
	where $\nvec{g}_p\in \BOmLCT$ and $\tilde{\nvec{g}}_{p}\rnc$ is reconstruction for given choice of $\{\Mn,\Mc,p\}$. The first two columns of \figref{fig:foldcountmap} show heatmaps with MSE for different $\{\Mn,\Mc\}$ and transforms: \FT, \FrFT, and \Fr. As expected, the heatmaps verify that  2nd\nobreakdash-order \MLSDQ achieves a smaller reconstruction MSE. Results show that optimal $\abs{\Mc} \in [0.1 \abs{\Mn}, 0.4 \abs{\Mn}]$ for the 1st\nobreakdash-order scheme and optimal $\abs{\Mc} \in [0.15 \abs{\Mn}, 0.6 \abs{\Mn}]$ for the 2nd\nobreakdash-order scheme for the majority of $\Mn$. Clearly, it is beneficial to use a larger window $\Mn$ to estimate folding locations and then use a smaller window $\Mc$ closer to the baseband interval to estimate the folding amplitudes. This is because frequency estimation is less sensitive to noise, and also choosing a smaller interval to estimate amplitudes implies smaller quantization noise.
	
	\bpara{E. Oversampling Ratio.} Here, we investigate the impact of the $\osrat$ and $\Mn$ on the reconstruction. For given $\osrat$, $\Mn$ and $l$-th order  quantizer, we define average MSE as:
	$
	\MSElhn{l} = \frac{1}{P}\sum\nolimits_{p=0}^{P-1}\MSEOp_{l}\rob{\nvec{g}_{p}, \tilde{\nvec{g}}_{p}\rhM}
	$ 
	where $\nvec{g}_p \in \BOmLCT$ for each trial $p$ and $\tilde{\nvec{g}}_{p}\rhM$ is the recovery for a given $\osrat$, $\Mn$ and $p$. We compare both schemes using the MSE difference
	$
	\MSEhndiff = \MSElhn{1} - \MSElhn{2}.
	$ 
	The third column of \figref{fig:foldcountmap} depicts $\MSEhndiff$ for an input signal with $M = 4$, $\abs{\Mc} = 0.4\abs{\Mn}$ and $K = 10$ folds for three different transforms: \FT in \figref{fig:foldcountmap} (c), \FrFT in \figref{fig:foldcountmap} (f), and \Fr in \figref{fig:foldcountmap} (i). In all the cases, the setup with 2nd-order \LSDQ\xspace achieves smaller reconstruction MSE for the majority of trials. The most significant performance improvement occurs in intervals with low oversampling (see \figref{fig:foldcountmap} (c), (f), (i)), where the 2nd-order scheme benefits from enhanced noise rejection.
	
	\section{Conclusion}
	In this paper, we present a \emph{generalized noise-shaping framework} with two key areas of generalization: (1) bandwidth and (2) dynamic range. By leveraging the flexibility of the Linear Canonical Transform (LCT) domain, we have expanded the applicability of \ob sampling to a wide range of signal classes beyond traditional Fourier domain assumptions. Our work extends the notion of bandwidth, enabling \ob sampling schemes to handle signals that may not be bandlimited in the Fourier domain but are bandlimited in other transform domains, such as the Fresnel, fractional Fourier, or Bargmann transforms. To address the overloading or saturation problem common in current \ob sampling methods, we incorporate the Unlimited Sensing Framework (USF). This novel approach separates key signal components—bandlimited input, modulo folds, and quantization noise—resulting in a new class of recovery algorithms optimized for transform-domain processing. Our approach outperforms existing time-domain methods, offering reduced oversampling requirements and enhanced robustness.
	
	The interplay of modulo and signum non-linearities at the core of our work opens several avenues for further research: (i) A deeper theoretical understanding of noise robustness in the transform domain is crucial for identifying the fundamental limits of dynamic range improvement. (ii) The joint utilization of time-domain and frequency-domain methods has the potential to significantly enhance system performance. (iii) Since oversampling is central to \ob sampling, the development of algorithms capable of efficiently handling large-scale data remains an important area of investigation.
	


\begin{thebibliography}{10}
\providecommand{\url}[1]{#1}
\csname url@samestyle\endcsname
\providecommand{\newblock}{\relax}
\providecommand{\bibinfo}[2]{#2}
\providecommand{\BIBentrySTDinterwordspacing}{\spaceskip=0pt\relax}
\providecommand{\BIBentryALTinterwordstretchfactor}{4}
\providecommand{\BIBentryALTinterwordspacing}{\spaceskip=\fontdimen2\font plus
\BIBentryALTinterwordstretchfactor\fontdimen3\font minus
  \fontdimen4\font\relax}
\providecommand{\BIBforeignlanguage}[2]{{%
\expandafter\ifx\csname l@#1\endcsname\relax
\typeout{** WARNING: IEEEtran.bst: No hyphenation pattern has been}%
\typeout{** loaded for the language `#1'. Using the pattern for}%
\typeout{** the default language instead.}%
\else
\language=\csname l@#1\endcsname
\fi
#2}}
\providecommand{\BIBdecl}{\relax}
\BIBdecl

\bibitem{Walden:2002:J}
R.~Walden, ``Analog-to-digital converter survey and analysis,'' \emph{{IEEE} J.
  Sel. Areas Commun.}, vol.~17, no.~4, pp. 539--550, 1999.

\bibitem{Thao:1994:J}
N.~Thao and M.~Vetterli, ``Deterministic analysis of oversampled {A/D}
  conversion and decoding improvement based on consistent estimates,''
  \emph{{IEEE} Trans. Sig. Proc.}, vol.~42, no.~3, pp. 519--531, Mar. 1994.

\bibitem{Verreault:2024:J}
A.~Verreault, P.-V. Cicek, and A.~Robichaud, ``{Oversampling} {ADC}: {A}
  {Review} of {Recent} {Design} {Trends},'' \emph{IEEE Access}, vol.~12, pp.
  121\,753--121\,779, Aug. 2024.

\bibitem{Schuchman:1964:J}
L.~Schuchman, ``{Dither} {Signals} and {Their} {Effect} on {Quantization}
  {Noise},'' \emph{{IEEE} Trans. Commun.}, vol.~12, no.~4, pp. 162--165, Dec.
  1964.

\bibitem{Inose:1963:J}
H.~Inose and Y.~Yasuda, ``A unity bit coding method by negative feedback,''
  \emph{Proc. IEEE}, vol.~51, no.~11, pp. 1524--1535, 1963.

\bibitem{Aziz:1996:J}
P.~Aziz, H.~Sorensen, and J.~vn~der Spiegel, ``An overview of sigma-delta
  converters,'' \emph{{IEEE} Signal Process. Mag.}, vol.~13, no.~1, pp. 61--84,
  Jan. 1996.

\bibitem{Daubechies:2003:J}
\BIBentryALTinterwordspacing
I.~Daubechies and R.~DeVore, ``{A}pproximating a {B}andlimited {F}unction
  {U}sing {V}ery {C}oarsely {Q}uantized {D}ata: {A} {F}amily of {S}table
  {S}igma-{D}elta {M}odulators of {A}rbitrary {O}rder,'' \emph{Ann. Math.},
  vol. 158, no.~2, pp. 679--710, 2003.
\BIBentrySTDinterwordspacing

\bibitem{Lazar:2004:J}
A.~Lazar and L.~Toth, ``{Perfect} {Recovery} and {Sensitivity} {Analysis} of
  {Time} {Encoded} {Bandlimited} {Signals},'' \emph{{IEEE} Trans. Circuits
  Syst. {II}}, vol.~51, no.~10, pp. 2060--2073, Oct. 2004.

\bibitem{Shamai:1994:J}
S.~Shamai, ``Information rates by oversampling the sign of a bandlimited
  process,'' \emph{{IEEE} Trans. Inf. Theory}, vol.~40, no.~4, pp. 1230--1236,
  Jul. 1994.

\bibitem{BarShalom:2002:J}
O.~Bar-Shalom and A.~Weiss, ``{DOA} estimation using one-bit quantized
  measurements,'' \emph{{IEEE} Trans. Aerosp. Electron. Syst.}, vol.~38, no.~3,
  pp. 868--884, Jul. 2002.

\bibitem{AlSafadi:2012:J}
E.~B. Al-Safadi and T.~Y. Al-Naffouri, ``{Peak} {Reduction} and {Clipping}
  {Mitigation} in {OFDM} by {Augmented} {Compressive} {Sensing},'' \emph{{IEEE}
  Trans. Sig. Proc.}, vol.~60, no.~7, pp. 3834--3839, Jul. 2012.

\bibitem{Olofsson:2005:J}
\BIBentryALTinterwordspacing
T.~Olofsson, ``{Deconvolution} and {Model}-{Based} {Restoration} of {Clipped}
  {Ultrasonic} {Signals},'' \emph{{IEEE} Trans. Instrum. Meas.}, vol.~54,
  no.~3, pp. 1235--1240, Jun. 2005.
\BIBentrySTDinterwordspacing

\bibitem{Sabharwal:2014:J}
A.~Sabharwal, P.~Schniter, D.~Guo, D.~W. Bliss, S.~Rangarajan, and R.~Wichman,
  ``{In}-{Band} {Full}-{Duplex} {Wireless}: {Challenges} and {Opportunities},''
  \emph{{IEEE} J. Sel. Areas Commun.}, vol.~32, no.~9, pp. 1637--1652, Sep.
  2014.

\bibitem{Zhang:2016:J}
\BIBentryALTinterwordspacing
J.~Zhang, J.~Hao, X.~Zhao, S.~Wang, L.~Zhao, W.~Wang, and Z.~Yao, ``Restoration
  of clipped seismic waveforms using projection onto convex sets method,''
  \emph{Nature Sci. Rep.}, vol.~6, no.~1, Dec. 2016.
\BIBentrySTDinterwordspacing

\bibitem{Bhandari:2017:C}
\BIBentryALTinterwordspacing
A.~Bhandari, F.~Krahmer, and R.~Raskar, ``On {Unlimited} {Sampling},'' in
  \emph{Intl. Conf. on Sampling Theory and Applications ({SampTA})}, Jul. 2017.
\BIBentrySTDinterwordspacing

\bibitem{Bhandari:2020:Ja}
------, ``{On} {Unlimited} {Sampling} and {Reconstruction},'' \emph{{IEEE}
  Trans. Sig. Proc.}, vol.~69, pp. 3827--3839, Dec. 2020.

\bibitem{Bhandari:2021:J}
A.~Bhandari, F.~Krahmer, and T.~Poskitt, ``{Unlimited} {Sampling} from {Theory}
  to {Practice}: {Fourier}-{Prony} {Recovery} and {Prototype} {ADC},''
  \emph{{IEEE} Trans. Sig. Proc.}, pp. 1131--1141, Sep. 2021.

\bibitem{Bhandari:2022:J}
A.~Bhandari, ``Back in the {US}-{SR}: {Unlimited} {Sampling} and {Sparse}
  {Super}-{Resolution} {With} {Its} {Hardware} {Validation},'' \emph{{IEEE}
  Signal Process. Lett.}, vol.~29, pp. 1047--1051, Mar. 2022.

\bibitem{Bhandari:2020:C}
A.~Bhandari and F.~Krahmer, ``{HDR} imaging from quantization noise,'' in
  \emph{{IEEE} Intl. Conf. on Image Processing ({ICIP})}, Oct. 2020, pp.
  101--105.

\bibitem{Romanov:2019:J}
E.~Romanov and O.~Ordentlich, ``Above the {Nyquist} {Rate}, {Modulo} {Folding}
  {Does} {Not} {Hurt},'' \emph{IEEE Signal Process. Lett.}, vol.~26, no.~8, pp.
  1167--1171, Aug. 2019.

\bibitem{Florescu:2022:J}
D.~Florescu, F.~Krahmer, and A.~Bhandari, ``The {Surprising} {Benefits} of
  {Hysteresis} in {Unlimited} {Sampling}: {Theory}, {Algorithms} and
  {Experiments},'' \emph{{IEEE} Trans. Sig. Proc.}, vol.~70, pp. 616--630, Jan.
  2022.

\bibitem{Florescu:2022:Ja}
D.~Florescu and A.~Bhandari, ``{Time} {Encoding} via {Unlimited} {Sampling}:
  {Theory}, {Algorithms} and {Hardware} {Validation},'' \emph{{IEEE} Trans.
  Sig. Proc.}, pp. 1--13, Sep. 2022.

\bibitem{Shtendel:2023:J}
G.~Shtendel, D.~Florescu, and A.~Bhandari, ``{Unlimited} {Sampling} of
  {Bandpass} {Signals}: {Computational} {Demodulation} via {Undersampling},''
  \emph{{IEEE} Trans. Sig. Proc.}, pp. 4134--4145, Sep. 2023.

\bibitem{Guo:2024:J}
R.~Guo, Y.~Zhu, and A.~Bhandari, ``{Sub-Nyquist} {USF} {Spectral} {Estimation}:
  {$K$} {Frequencies} {With} {$6K+4$} {Modulo} {Samples},'' \emph{{IEEE} Trans.
  Sig. Proc.}, vol.~72, pp. 5065--5076, 2024.

\bibitem{Zhu:2024:C}
Y.~Zhu, R.~Guo, P.~Zhang, and A.~Bhandari, ``{Frequency} {Estimation} via
  {Sub}-{Nyquist} {Unlimited} {Sampling},'' in \emph{{IEEE} Intl. Conf. on
  Acoustics, Speech and Signal Processing (ICASSP)}, Apr. 2024.

\bibitem{Mulleti:2023:J}
S.~Mulleti, E.~Reznitskiy, S.~Savariego, M.~Namer, N.~Glazer, and Y.~C. Eldar,
  ``A hardware prototype of wideband high‐dynamic range analog‐to‐digital
  converter,'' \emph{IET Circuits, Devices \& Systems}, vol.~17, no.~4, pp.
  181--192, Jun. 2023.

\bibitem{Ordentlich:2018:J}
O.~Ordentlich, G.~Tabak, P.~K. Hanumolu, A.~C. Singer, and G.~W. Wornell, ``{A}
  {Modulo}-{Based} {Architecture} for {Analog}-to-{Digital} {Conversion},''
  \emph{IEEE Journal of Selected Topics in Signal Processing}, vol.~12, no.~5,
  pp. 825--840, Oct. 2018.

\bibitem{Graf:2019:C}
O.~Graf, A.~Bhandari, and F.~Krahmer, ``One-bit {Unlimited} {Sampling},'' in
  \emph{{IEEE} Intl. Conf. on Acoustics, Speech and Signal Processing
  (ICASSP)}, May 2019, pp. 5102--5106.

\bibitem{Florescu:2022:Ca}
D.~Florescu and A.~Bhandari, ``Modulo {Event}-driven {Sampling}: {System}
  {Identification} and {Hardware} {Experiments},'' in \emph{{IEEE} Intl. Conf.
  on Acoustics, Speech and Signal Processing (ICASSP)}, May 2022, pp.
  5747--5751.

\bibitem{Eamaz:2024:J}
A.~Eamaz, K.~V. Mishra, F.~Yeganegi, and M.~Soltanalian, ``{UNO}: {Unlimited}
  {Sampling} {Meets} {One}-{Bit} {Quantization},'' \emph{{IEEE} Trans. Sig.
  Proc.}, vol.~72, pp. 997--1014, 2024.

\bibitem{Eamaz:2023:C}
------, ``{Unlimited} {Sampling} via {One}-{Bit} {Quantization},'' in
  \emph{Intl. Conf. on Sampling Theory and Applications (SampTA)}.\hskip 1em
  plus 0.5em minus 0.4em\relax IEEE, Jul. 2023.

\bibitem{Eamaz:2023:Ca}
A.~Eamaz, F.~Yeganegi, K.~V. Mishra, and M.~Soltanalian, ``Unlimited {Sampling}
  of {FRI} {Signals} with {Dithered} {One}-{Bit} {Quantization},'' in
  \emph{57th Asilomar Conf. on Signals, Systems, and Computers}.\hskip 1em plus
  0.5em minus 0.4em\relax IEEE, Oct. 2023.

\bibitem{Florescu:2021:C}
D.~Florescu, F.~Krahmer, and A.~Bhandari, ``Event-{Driven} {Modulo}
  {Sampling},'' in \emph{{IEEE} Intl. Conf. on Acoustics, Speech and Signal
  Processing (ICASSP)}, Jun. 2021.

\bibitem{Florescu:2022:C}
D.~Florescu and A.~Bhandari, ``{Unlimited} {Sampling} with {Local}
  {Averages},'' in \emph{{IEEE} Intl. Conf. on Acoustics, Speech and Signal
  Processing (ICASSP)}, May 2022, pp. 5742--5746.

\bibitem{Gori:1981:J}
F.~Gori, ``Fresnel transform and sampling theorem,'' \emph{Optics
  Communications}, vol.~39, no.~5, pp. 293--297, Nov. 1981.

\bibitem{Martone:2001:J}
M.~Martone, ``A multicarrier system based on the fractional {Fourier} transform
  for time-frequency-selective channels,'' \emph{{IEEE} Trans. Commun.},
  vol.~49, no.~6, pp. 1011--1020, Jun. 2001.

\bibitem{Rou:2024:J}
H.~S. Rou, G.~T.~F. de~Abreu, J.~Choi, D.~González~G., M.~Kountouris, Y.~L.
  Guan, and O.~Gonsa, ``{From} {Orthogonal} {Time}–{Frequency} {Space} to
  {Affine} {Frequency}-{Division} {Multiplexing}: {A} {comparative} study of
  next-generation waveforms for integrated sensing and communications in doubly
  dispersive channels,'' \emph{{IEEE} Signal Process. Mag.}, vol.~41, no.~5,
  pp. 71--86, Sep. 2024.

\bibitem{Bhandari:2020:FrFT}
A.~Bhandari, O.~Graf, F.~Krahmer, and A.~Zayed, ``One-{Bit} {Sampling} in
  {Fractional} {Fourier} {Domain},'' in \emph{Proc. {IEEE} Int. Conf. Acoust.,
  Speech, Sig. Proc.}\hskip 1em plus 0.5em minus 0.4em\relax IEEE, May 2020.

\bibitem{Moshinsky:1971:J}
M.~Moshinsky and C.~Quesne, ``{Linear} {Canonical} {Transformations} and
  {Their} {Unitary} {Representations},'' \emph{J. Math. Phys.}, vol.~12, no.~8,
  pp. 1772--1780, Aug. 1971.

\bibitem{Barshan:1997:J}
B.~Barshan, M.~Kutay, and H.~M. Ozaktas, ``Optimal filtering with linear
  canonical transformations,'' \emph{Optics Communications}, vol. 135, no.
  1–3, pp. 32--36, Feb. 1997.

\bibitem{Sharma:2009:J}
K.~Sharma, ``{Approximate} {Signal} {Reconstruction} {Using} {Nonuniform}
  {Samples} in {Fractional} {Fourier} and {Linear} {Canonical} {Transform}
  {Domains},'' \emph{{IEEE} Trans. Sig. Proc.}, vol.~57, no.~11, pp.
  4573--4578, Nov. 2009.

\bibitem{Xu:2017:J}
L.~Xu, R.~Tao, and F.~Zhang, ``{Multichannel} {Consistent} {Sampling} and
  {Reconstruction} {Associated} {With} {Linear} {Canonical} {Transform},''
  \emph{{IEEE} Signal Process. Lett.}, vol.~24, no.~5, pp. 658--662, May 2017.

\bibitem{Zayed:2018:J}
A.~I. Zayed, ``{Sampling} of {Signals} {Bandlimited} to a {Disc} in the
  {Linear} {Canonical} {Transform} {Domain},'' \emph{{IEEE} Signal Process.
  Lett.}, vol.~25, no.~12, pp. 1765--1769, Dec. 2018.

\bibitem{Dahlen:2017:J}
D.~Dahlen, R.~Wilcox, and W.~Leemans, ``{Modeling} {Herriott} cells using the
  linear canonical transform,'' \emph{Appl. Optics}, vol.~56, no.~2, p. 267,
  Jan. 2017.

\bibitem{Mohammed:2024:J}
E.~A. Mohammed and I.~M. Qasim, ``Optical double-image cryptosystem based on a
  joint transform correlator in a linear canonical domain,'' \emph{Applied
  Optics}, vol.~63, no.~22, p. 5941, Jul. 2024.

\bibitem{Sahin:1998:J}
A.~Sahin, H.~M. Ozaktas, and D.~Mendlovic, ``Optical implementations of
  two-dimensional fractional fourier transforms and linear canonical transforms
  with arbitrary parameters,'' \emph{Applied Optics}, vol.~37, no.~11, p. 2130,
  Apr. 1998.

\bibitem{Bemani:2023:J}
A.~Bemani, N.~Ksairi, and M.~Kountouris, ``{Affine} {Frequency} {Division}
  {Multiplexing} for {Next} {Generation} {Wireless} {Communications},''
  \emph{{IEEE} Trans. Wireless Commun.}, vol.~22, no.~11, pp. 8214--8229, Nov.
  2023.

\bibitem{Luo:2024:J}
Q.~Luo, P.~Xiao, Z.~Liu, Z.~Wan, N.~Thomos, Z.~Gao, and Z.~He, ``{AFDM}-{SCMA}:
  {A} {Promising} {Waveform} for {Massive} {Connectivity} {Over} {High}
  {Mobility} {Channels},'' \emph{{IEEE} Trans. Wireless Commun.}, vol.~23,
  no.~10, pp. 14\,421--14\,436, Oct. 2024.

\bibitem{Bhandari:2019:J}
A.~Bhandari and A.~I. Zayed, ``Shift-invariant and sampling spaces associated
  with the special affine {Fourier} transform,'' \emph{Appl. Comput. Harmon.
  A.}, vol.~47, no.~1, pp. 30--52, Jul. 2019.

\bibitem{Souvorov:2006:J}
A.~Souvorov, T.~Ishikawa, and A.~Kuyumchyan, ``Multiresolution phase retrieval
  in the {Fresnel} region by use of wavelet transform,'' \emph{J. Opt. Soc.
  Amer. A}, vol.~23, no.~2, p. 279, Feb. 2006.

\bibitem{Wolf:2007:J}
K.~B. Wolf and G.~Krötzsch, ``Geometry and dynamics in the {Fresnel}
  transforms of discrete systems,'' \emph{J. Opt. Soc. Amer. A}, vol.~24,
  no.~9, p. 2568, 2007.

\bibitem{Zayed:2021:J}
A.~I. Zayed, ``{Sampling} theorem for two dimensional fractional {Fourier}
  transform,'' \emph{Signal Process.}, vol. 181, p. 107902, Apr. 2021.

\bibitem{Liebling:2003:J}
M.~Liebling, T.~Blu, and M.~Unser, ``Fresnelets: new multiresolution wavelet
  bases for digital holography,'' \emph{{IEEE} Trans. Image Proc.}, vol.~12,
  no.~1, pp. 29--43, Jan. 2003.

\bibitem{Bhandari:2012:J}
A.~Bhandari and A.~I. Zayed, ``{Shift}-{Invariant} and {Sampling} {Spaces}
  {Associated} {With} the {Fractional} {Fourier} {Transform} {Domain},''
  \emph{{IEEE} Trans. Sig. Proc.}, vol.~60, no.~4, pp. 1627--1637, Apr. 2012.

\bibitem{Zhang:2023:C}
P.~Zhang and A.~Bhandari, ``Unlimited {Sampling} in {Phase} {Space},'' in
  \emph{{IEEE} Intl. Conf. on Acoustics, Speech and Signal Processing
  (ICASSP)}, Jun. 2023.

\bibitem{Stoica:1997:B}
P.~G. Stoica and R.~L. Moses, \emph{Introduction to spectral analysis},
  P.~Stoica, Ed.\hskip 1em plus 0.5em minus 0.4em\relax Upper Saddle River, NJ
  [u.a.]: Prentice Hall, 1997, literaturverz. S. 299 - 308.

\bibitem{Hua:1990:J}
Y.~Hua and T.~Sarkar, ``Matrix pencil method for estimating parameters of
  exponentially damped/undamped sinusoids in noise,'' \emph{{IEEE} Trans.
  Acoust., Speech, Signal Process.}, vol.~38, no.~5, pp. 814--824, May 1990.

\end{thebibliography}
\end{document}